\documentclass[10pt]{extarticle}
\usepackage[utf8]{inputenc}
\usepackage[margin=1in]{geometry}
\usepackage{cite}
\usepackage{amsmath}
\usepackage{amsthm}
\usepackage{amssymb}
\usepackage{enumitem}
\usepackage{graphicx}
\usepackage{comment}
\usepackage{hyperref}
\usepackage{xcolor}
\usepackage{bm}
\usepackage{url}
\usepackage[normalem]{ulem}
\usepackage[parfill]{parskip}
\usepackage{thmtools}
\usepackage{thm-restate}
\usepackage{etoolbox}
\newtoggle{arxiv}
\toggletrue{arxiv}

\newtheorem{theorem}{Theorem}[section]

\newtheorem{corollary}{Corollary}[theorem]
\newtheorem{lemma}{Lemma}[section]

\newtheorem{definition}{Definition}[section]
\def\Adv{\ensuremath{\mathcal{A}}}

\def \CharSymbols {\ensuremath{\{\perp, 0, 1\}}}

\newcommand{\Geom}{\mathsf{geom}}
\newcommand{\Unheard}{\mathsf{Unheard}}
\newcommand{\CompressedUnheard}{\mathsf{CompressedUnheard}}
\newcommand{\CharString}{\mathsf{CharString}}
\newcommand{\CompressedCharString}{\mathsf{CompressedCharString}}
\newcommand{\LeaderString}{\mathsf{LeaderString}}
\newcommand{\Advantage}{\mathsf{Advantage}}
\newcommand{\CompressedAdvantage}{\mathsf{CompressedAdvantage}}
\newcommand{\Margin}{\mathsf{Margin}}
\newcommand{\CompressedMargin}{\mathsf{CompressedMargin}}
\newcommand{\LatestHeard}{\mathsf{LatestHeard}}
\newcommand{\Reach}{\mathsf{Reach}}
\newcommand{\CompressedReach}{\mathsf{CompressedReach}}

\title{The Longest-Chain Protocol Under Random Delays}
\author{Suryanarayana Sankagiri$^*$, Shreyas Gandlur$^\dagger$, Bruce Hajek$^*$ \\
$^*$ University of Illinois at Urbana-Champaign \\
$^\dagger$ Princeton University}
\date{\today}

\begin{document}
\maketitle
 
\begin{abstract}
In the field of distributed consensus and blockchains, the synchronous communication model assumes that all messages between honest parties are delayed at most by a known constant $\Delta$. Recent literature establishes that the longest-chain blockchain protocol is secure under the synchronous model. However, for a fixed mining rate, the security guarantees degrade with $\Delta$. We analyze the performance of the longest-chain protocol under the assumption that the communication delays are random, independent, and identically distributed. This communication model allows for distributions with unbounded support and is a strict generalization of the synchronous model. We provide safety and liveness guarantees with simple, explicit bounds on the failure probabilities. These bounds hold for infinite-horizon executions and decay exponentially with the security parameter. In particular, we show that the longest-chain protocol has good security guarantees when delays are sporadically large and possibly unbounded, which is reflective of real-world network conditions.
\end{abstract}

\section{Introduction}\label{sec:introduction}

Over the past ten years, blockchains have generated tremendous interest by enabling decentralized payment systems. The term blockchain, first introduced by Satoshi Nakamoto in his design of Bitcoin \cite{nakamotobitcoin}, refers to the distributed data structure at the heart of these systems. Consensus protocols are used to ensure the consistency of blockchains among different parties. Although the data structure itself has remained fairly standard, associated consensus protocols have proliferated. We refer the reader to \cite{bano2019sok} and \cite{garay2020sok} for surveys on different blockchain consensus protocols. In this work, we restrict our attention to the {\em longest-chain protocol}, or Nakamoto consensus, which forms the backbone of various popular cryptocurrencies like Bitcoin and Ethereum. 

Many papers have formally studied this protocol's security under under a variety of modeling assumptions. These modeling assumptions vary, among other things, with respect to the nature of the leader election mechanism (Proof of Work-PoW \cite{garay2015bitcoin, pass2017analysis} versus Proof of Stake-PoS \cite{kiayias2017ouroboros, pass2017sleepy}), and the timing assumptions (continuous time \cite{li2020close, dembo2020everything} versus discrete time \cite{blum2020combinatorics, gazi2020tight}). Notwithstanding these modeling differences, some basic principles behind the security of the protocol have emerged. 

One fundamental principle is that the longest chain protocol is secure in the synchronous network model, under sufficient honest representation. In this model, a message sent at time $\tau$ will be delivered by time $\tau + \Delta$, where $\Delta$ is a system parameter. Under this assumption, \cite{dembo2020everything, gazi2020tight} show that the protocol is secure if and only if 
\begin{equation}\label{eq:adv_threshold}
    \beta < \frac{1-\beta}{1 + (1-\beta)f \Delta}
\end{equation}
where $\beta$ is the fraction of adversarial power and $f$ is the mining rate (or block production rate). The term $1/(1 + (1-\beta) f \Delta)$ can be thought of as a {\em discount factor} in the honest power, capturing the effect of the message delays. Succinctly put, \eqref{eq:adv_threshold} states that the security threshold of the protocol degrades with $f\Delta$.

A second principle is that when the tuple $(\beta, f, \Delta)$ satisfies \eqref{eq:adv_threshold}, the protocol satisfies both safety (all honest parties have consistent chains, except for the last few blocks) and liveness (new honest blocks are included in all parties' chains at a regular rate) security properties with high probability. In fact, the probability that these properties are violated decreases exponentially with a parameter $k$. Recent works (e.g., \cite{li2020close, blum2020combinatorics}) state security properties in a form such that the probability of violations remains negligible even for infinite horizon executions. The security statements in this form are more general, and imply bounds for the statements given in other works such as \cite{garay2015bitcoin, kiayias2017ouroboros}. The aforementioned principles hold for both PoW and PoS versions of the protocol.

In real-world conditions, worst-case message delays may be much larger than typical delays. Therefore, the longest-chain protocol may have better security guarantees than those suggested by analysis which sets $\Delta$ to the maximum possible delay. The use of the random delay model, as proposed in this paper, formalizes this intuition.

\subsection{Our Contributions}\label{sec:our_contributions}
This paper studies the security of the longest chain protocol in a network with random, possibly unbounded, delays. Briefly, each peer-to-peer communication is subject to an independent and identically distributed (i.i.d.) delay. Thus, different recipients of a broadcast may receive the message at different times. This communication model is a generalization of the synchronous model, and has not been studied in prior work on blockchain security.

Drawing inspiration from statistical physics, this paper states and distinguishes between two forms of security properties: \textit{intensive} and \textit{extensive}. Intensive security properties capture the security of localized portions of blockchains, whereas extensive security properties provide global security guarantees. Prior works typically state properties in only one of these forms; those works that state properties in both forms do not formally distinguish them. We show that guarantees for the intensive forms imply guarantees for the extensive forms.

Our main result, Theorem \ref{thm:main}, states that the longest-chain protocol satisfies the settlement and chain quality properties in the random delay model, except with probability that decays exponentially in a wait-time (or security parameter) $k$. These properties are intensive forms of safety and liveness, and pertain to an infinite-horizon execution. We provide explicit error bounds. As in the synchronous model, the security guarantees hold under appropriate bounds on the adversarial power and the mining rate.

Our work highlights the dual role of communication delays: these delays have both global and local effects. Delays in messages from past leaders to future leaders have global effect: they influence the growth of the longest chain and impact the security of all honest parties. We generalize the analysis tools developed in the Ouroboros line of papers \cite{kiayias2017ouroboros, david2018ouroboros, badertscher2018ouroboros, blum2020combinatorics} to handle these delays (e.g., see Section \ref{sec:special_honest_slots}, which describes a generalization of characteristic strings). In contrast, delays in messages from leaders to a given honest observer $h$ have local impact: they affect the length of the chain held by $h$. We define a new local metric called $\Unheard_h$ (see Section \ref{sec:unheard}) to handle these delays. Theorem \ref{thm:main} reflects this dual role of delays. The error bounds of the security statements include two terms: one is a bound on atypical behavior of the characteristic string; the other is a bound on atypical behavior of $\Unheard_h$ for every honest party $h$. Note that a given party may be a leader in one context and an observer in another. 

\subsection{Comparison with Prior Work}\label{sec:related_work}
\paragraph{Communication Model} We compare the random delay model of this work to the \textit{partially synchronous model} \cite{dwork1988consensus} and the \textit{sleepy model} \cite{pass2017sleepy} of communication. The partially synchronous model assumes that message delays are unbounded until an adversarially chosen time $T^*$, and are bounded thereafter. In this model, the longest-chain protocol is secure only after $O(T^*)$ time, as shown in \cite{neu2020ebb}. In the sleepy model, the adversary can put an honest party to sleep for an arbitrary period of time. Sleepy honest parties have unbounded delay (equivalently, they do not communicate), while the awake parties have bounded delays. The longest-chain protocol is secure in the sleepy model, provided the fraction of awake honest parties exceeds the fraction of corrupt parties (see \cite{pass2017sleepy}).

In both models, unbounded delays are localized--to select period(s) of time in the partially synchronous model and to select parties in the sleepy model. In comparison, the random delay model conveys a more homogeneous network setting.  Here, message delays from any honest party at any time may be large--across parties and time simultaneously. Although each of these models describe communication settings with sporadic large delays, they capture different facets of sub-optimal network behavior. Studying the same protocol in different models provides a better understanding of its real-world performance.

Two other works \cite{fanti2019barracuda, gopalan2020stability} study the longest-chain protocol under random, unbounded delay. These works model the network as a graph of inter-connected nodes, and assume that delays between two neighboring nodes in the network are exponentially distributed and i.i.d. However, neither of these works analyze an adversary trying to disrupt the security of the protocol. In this work, we model point-to-point communication instead of communication over a graph. We also allow for general delay distributions.

\paragraph{Security Analysis} Our statement, and proof, of the settlement (safety) property draw inspiration from Blum et al. \cite{blum2020combinatorics}. Blum et al. show that PoS longest-chain protocols satisfy safety with an error probability that decays exponentially in the wait-time $k$. Moreover, the proof in \cite{blum2020combinatorics} yields explicit expressions for the constants in the error bounds. For PoW models, \cite{li2020close} provides explicit, exponentially decaying error bounds for intensive security properties.

The analysis of \cite{blum2020combinatorics}, which is for the special case $\Delta = 0$, can be extended to any constant $\Delta$ as shown in \cite{david2018ouroboros, badertscher2018ouroboros}. Similarly, we adapt the analysis to the random delay model. We generalize the notion of $\Delta$-isolated slots in \cite{david2018ouroboros, badertscher2018ouroboros} to that of \textit{special honest slots}, retaining the property that blocks from these slots must be at different heights. The statement and proof of the chain quality property in this work is inspired by \cite{badertscher2018ouroboros}. Our work focuses on the intensive form of this property, which also applies to the extensive form, while the analysis in \cite{badertscher2018ouroboros} is only for the extensive property.

The works of Dembo et al. \cite{dembo2020everything} and Gazi et al. \cite{gazi2020tight} give a tight characterization of the security regime of the longest chain protocol via \eqref{eq:adv_threshold}. The security threshold is obtained by comparing the growth rate of the adversarial chain with that of the honest tree (the private attack). Obtaining an expression for the growth rate of the honest tree in the random delay model, and extending the results of \cite{dembo2020everything, gazi2020tight} to this model are directions for future research.

\section{The System Model}\label{sec:model}

\subsection{Preliminaries}\label{sec:preliminaries}
The protocol proceeds in discrete time slots that are indexed by $\mathbb{N}$ and runs for an infinite duration. We assume that clocks of all parties are perfectly synchronized. Blocks are treated as abstract data structures containing an integer timestamp, a hash pointer to a parent block with a smaller timestamp, a cryptographic signature of the block's proposer, some transactions and other relevant information. A special genesis block, with timestamp $0$ and no parent, is known to all parties at the start of the protocol.  We assume the existence of a leader election mechanism which selects a subset of parties in each time slot to be leaders for that slot. Only leaders can propose blocks with the corresponding timestamp. This mechanism is an abstraction of the mining process in PoW systems or the leader election protocol in PoS systems.

\paragraph{Parties in the protocol}
The parties in the protocol comprise of honest ones and a single adversary $\mathcal{A}$. (Replacing all corrupt parties by a single one is done for simplicity). The set of honest parties is represented by $\mathcal{H}$ and may be finite or infinite. Arbitrary honest parties are denoted by $h, h_1, h_2,$ etc. In our model, the adversary can never corrupt an honest party and the honest parties never go offline. Honest parties follow the longest-chain protocol, while the adversary can deviate from the protocol arbitrarily. The precise difference between honest and adversarial actions are given in Section \ref{sec:slot}.

\paragraph{Blockchains} From any block, a unique sequence of blocks leading up to the genesis block can be identified via the hash pointers. We call this sequence a \textit{blockchain}, or simply a \textit{chain}. The convention is that the genesis block is the first block of the chain, and the terminating block is called the \textit{tip}. The timestamps of blocks in a blockchain must strictly increase, going from the genesis to the tip. At any given slot, honest parties store a single chain in their memory. We use $\mathcal{C}^h_i$ to denote the chain held by an honest party $h$ at (the end of) slot $i$. We use $\mathcal{C}[i_1:i_2]$ to represent the portion of a chain $\mathcal{C}$ consisting of blocks with timestamps in the interval $\{i_1, \ldots, i_2\}$.

\paragraph{Blocktrees} The set of all blocks generated up to a given slot $i$ forms a directed tree. Let $\mathcal{F}_i$ be the directed graph $(V, E)$, where $V$ is the set of blocks generated up to slot $i$ and $E$ is the set of parent-child block pairs. These edges point from parent to child, in the opposite direction of the hash pointers. The genesis block is the root of the tree, with no parent. In addition, the timestamp of block $v$ is denoted by $\ell(v)$. Every blockchain $\mathcal{C}^h_i$ is a directed path in $\mathcal{F}_i$ that begins at the genesis block and ends at any other block. $\mathcal{F}_i$ includes blocks held privately by the adversary.

\subsection{Details of a Slot}\label{sec:slot}
Within a slot, the following events occur in the given order. This describes the prescribed honest protocol, and also specifies the adversary's powers.
\begin{itemize}
    \item (\textbf{Leader Election Phase}) All parties learn the slot leaders through the leader election mechanism.
    \item (\textbf{Honest Send Phase}) Honest leaders create a new block, append it to their chain, and broadcast this new chain to all parties. The communication network assigns random delay to each point-to-point message.
    \item (\textbf{Adversarial Send Phase}) $\mathcal{A}$ receives all chains sent (if any) in the Honest Send Phase, along with their respective message delays. $\mathcal{A}$ may then create some new blocks with timestamps of any slot for which it was elected a leader and may create multiple blocks with the same timestamp. It sends each new block (along with the preceding blockchain) to an  arbitrary subset of honest parties. 
    \item (\textbf{Deliver Phase}) Messages from honest parties slated for delivery in the current slot and $\mathcal{A}$'s messages from the current slot are delivered to the appropriate honest parties. $\mathcal{A}$ can also choose to deliver any honest messages ahead of schedule.
    \item (\textbf{Adopt Phase})  Each honest party updates its chain if it receives any chain strictly longer than the one it holds. If an honest party receives multiple longer chains, it chooses the longest one with $\mathcal{A}$ breaking any ties.
\end{itemize}

\subsection{Leader Election}\label{sec:leader_election}
We model the leader election mechanism such that the sets of leaders in different slots are independent and identically distributed subsets of $\mathcal{H}\cup \{ \mathcal{A}\}.$ For example, the leader election process in the first few slots may be: $\{h_1\}$, $\emptyset$, $\{\mathcal{A}\}$, $\{h_2, h_3\}$, $\emptyset$, $\{h_4, \mathcal{A}\}$, $\{h_1, h_5, \mathcal{A}\}.$ Let $\mathcal{L}_s$ denote the set of leaders in slot $s$.   The adversary cannot influence the leader election mechanism.
Let $A_s=1$ if $\mathcal{A} \in \mathcal{L}_s$ and $A_s=0$ otherwise, and let $N_s$ be the number of honest leaders in slot $\mathcal{L}_s$. Note that $(N_s, A_s)$ may have any possible joint distribution, but the process $\{(N_s, A_s)\}_{s \geq 1}$ is i.i.d. Let $(N, A)$ denote a representative random tuple of the aforementioned process. Define
\begin{itemize}
    \item $f \triangleq \mathbb{P}(A + N > 0)$. $f$ is the probability of a \textit{non-empty slot}, i.e., a slot with one or more leaders. In a sense, it is the mining rate of the protocol.
    \item $\alpha \triangleq \mathbb{P}(N = 1 \ \& \ A = 0 \, \vert \, A + N > 0)$. $\alpha$ is the probability of having a unique honest leader in a slot, given that the slot is a non-empty slot.
\end{itemize}

\subsection{The Communication Model}\label{sec:communication_network}
We now describe our model of the communication network, the \textit{random delay model}. Every message sent by one honest party to another is subject to a random delay, which can take any value in $\mathbb{Z}_+$. Note that a broadcast is a set of different point-to-point messages, each of which is subject to an independent delay. We adopt the convention that the minimum possible delay is zero; in this case, a message sent in a time slot is received by the end of that slot. The delays of different messages are i.i.d.; let $\Delta$ denote a random variable with this distribution, called the \textit{delay distribution}. The synchronous model is a special case with a constant $\Delta$.

For technical reasons, we require that the delay distribution has a \textit{non-decreasing failure rate function}. The failure rate function for the delay distribution $\Delta$, is defined as 
\[ \textsf{Failure Rate}(s) = \begin{cases}\mathbb{P}(\Delta = s \vert \Delta \geq s) & \text{if } \mathbb{P}(\Delta \geq s)  > 0\\
1 & \text{if } \mathbb{P}(\Delta \geq s)  = 0\end{cases}\]
A geometric random variable has a constant failure rate. A constant $\Delta$ has a failure rate function that is $0$ up to the constant and $1$ thereafter. Therefore, they are both admissible in our model. A consequence of a non-decreasing failure rate is that, for all $i \geq 0$ and $s \geq 0$ such that $\mathbb{P}(\Delta \geq s) > 0$, $\mathbb{P}(\Delta \geq s + i \vert \Delta \geq s) \leq \mathbb{P}(\Delta \geq i)$.

Given the power of the adversary to deliver honest messages earlier than scheduled, a system with a given delay distribution $\Delta$ can be subsumed by a system that has a different delay distribution $\Tilde{\Delta}$, provided the latter stochastically dominates the former. If $\Tilde{\Delta}$ satisfies the non-decreasing failure rate restriction, guarantees for a system with delay $\Delta$ can be given in terms of the distribution $\Tilde{\Delta}$.

The non-decreasing failure rate restriction is not a fundamental limitation of the model, but rather of the method of analysis. One technique of removing this restriction is to assume a slightly different form of the leader election process. This is described next.

\subsection{One-Time Leader Model} \label{sec:infinite_users}
Consider an alternate model of the leader election process in which each honest party can be chosen as a leader at most once. We call this model the \textit{one-time leader model} to distinguish it from the \textit{i.i.d. leader model} described in Section \ref{sec:leader_election}. Let the set of honest parties $\mathcal{H}$ be divided into two groups, leaders $\mathcal{M}$ and observers $\mathcal{O}$. The set of miners is countably infinite and are indexed $m_1, m_2, \ldots$. The set of observers may be finite or infinite. Leaders are chosen among parties in $\mathcal{M}$ in the order of their indexing. In this model, the sets of leaders in each slot is no longer independent. However, the tuples $\{(N_s, A_s)\}_{s \geq 1}$ are i.i.d., where $N_s$ and $A_s$ have the same interpretation as before. The parameters $f$ and $\alpha$ are also defined in the same manner as before.

This alternate leader election model allows us to extend the security analysis to any delay distribution. In particular, $\Delta$ can now be infinity with some probability. A delay of infinity for a message implies that the adversary can choose to deliver the message at any time of its choice, or never at all.

\section{The Desired Security Properties}\label{sec:desired_properties}

The security properties defined in this section, and the guarantees for them hold for both models introduced in Section \ref{sec:model}: the i.i.d. leader model with delays having a non-decreasing failure rate function and the one-time leader model with general delay distributions. Each security property refers to a desirable condition over an \textit{execution}. Formally, an execution of the protocol refers to a particular instantiation of the random components (i.e., leader election and communication delays) and the actions of the adversary. Whether a certain property holds or not in an execution depends on both these factors. The adversary's actions can be arbitrary and our theorems are stated for the worst-case scenario of all possible adversarial actions.

\subsection{Property Definitions}\label{sec:property_definitions}
We first define the settlement property, which is an intensive form of safety (see \cite{blum2020combinatorics} for the original definition).
\begin{definition}[Settlement]\label{def:settlement}
    In an execution, the \emph{settlement property with parameters $s, k \in \mathbb{N}$ and $\mathcal{I} \subseteq \mathcal{H}$} holds if, for any pair of honest parties $h_1, h_2 \in \mathcal{I}$ and slots $i_1, i_2$ such that $s + k \leq i_1 \leq i_2$, it holds that $\mathcal{C}^{h_1}_{i_1}[1:s] = \mathcal{C}^{h_2}_{i_2}[1:s]$. 
\end{definition}
We refer to the settlement property with parameters $s, k$, and $\mathcal{I}$ as the ($s, k, \mathcal{I}$)-settlement property for brevity. We use a similar convention for other properties too. The ($s, k, \mathcal{I}$)-settlement property, roughly speaking, means that parties in $\mathcal{I}$ will agree on the order of blocks mined up to slot $s$ after $k$ more slots.
We now state the common prefix property, an extensive form of safety.
\begin{definition}[Common Prefix]\label{def:common_prefix}
    In an execution, the \emph{common prefix property with parameters $T, k \in \mathbb{N}$ and $\mathcal{I} \subseteq \mathcal{H}$} holds if, for any pair of honest players $h_1, h_2 \in \mathcal{I}$ and slots $s, i_1, i_2$ such that $s \leq T$ and $s + k \leq i_1 \leq i_2$, it holds that $\mathcal{C}^{h_1}_{i_1}[1:s] = \mathcal{C}^{h_2}_{i_2}[1:s]$.
\end{definition}
The intensive and extensive forms of safety have a subtle difference, which we illustrate with an example. Let $T$ be some large number. The ($T, k, \mathcal{I}$)-settlement property means the parties in $\mathcal{I}$ agree forever after slot $T+k$ about the chain up to slot $T$. This immediately implies that all parties in $\mathcal{I}$ agree forever about the chain up to slot $s$, after slot $T+k$, for any $s \leq T$. This does not, however, imply that all parties in $\mathcal{I}$ agree forever about the chain up to time $s$, after slot $s+k$, for all $s$ with $s \leq T$. This latter statement is captured by the extensive form given by the common prefix property. Formally, the intensive and extensive forms of common prefix are related as follows:
\begin{lemma}\label{lem:settlement_CP}
Fix a set of honest players $\mathcal{I} \subseteq \mathcal{H}$ and parameters $T, k \in \mathbb{N}$. If the settlement property holds with parameters $s, k$ and $\mathcal{I}$ for all $s \leq T$, then the common prefix property holds with parameters $T, k$ and $\mathcal{I}$. 
\end{lemma}
\begin{proof}
    Pick any pair of honest users, $h_1, h_2 \in \mathcal{I}$ and any slot $s \leq T$. Pick any $i_1, i_2$ satisfying $s+k \leq i_1 \leq i_2$. Consider the chains held by $h_1, h_2$ at slots $i_1, i_2$ respectively: $\mathcal{C}^{h_1}_{i_1}, \mathcal{C}^{h_2}_{i_2}$. We wish to show that $\mathcal{C}^{h_1}_{i_1}[1:s] = \mathcal{C}^{h_2}_{i_2}[1:s]$. But this follows from the settlement property with parameters $s, k$ and $\mathcal{I}$.
\end{proof}

We next state the chain quality property, first in its intensive form and then in its extensive form.
\begin{definition}[Intensive Chain Quality]\label{def:chain_quality_intensive}
    In an execution, the \emph{intensive chain quality property} with parameters $\mu \in (0, 1)$, $s, k \in \mathbb{N}$ and $\mathcal{I} \subseteq \mathcal{H}$ holds if, for any honest players $h \in \mathcal{I}$ and slot $i \geq s + k$, $\mathcal{C}^{h}_{i}[s+1:s+k]$ contains greater than $k f \mu $ honestly mined blocks. 
\end{definition}
\begin{definition}[Extensive Chain Quality]\label{def:chain_quality_extensive}
    In an execution, the \emph{extensive chain quality property} with parameters $\mu \in (0, 1)$, $T, k \in \mathbb{N}$ and $\mathcal{I} \subseteq \mathcal{H}$ holds if, for any honest players $h \in \mathcal{I}$ and slots $s, i$ such that $s \leq T$ and $i \geq s + k$, $\mathcal{C}^{h}_{i}[s+1:s+k]$ contains greater than $k f \mu $ honestly mined blocks. 
\end{definition}
The relation between the intensive and extensive versions of chain quality parallels that between the settlement and common prefix property noted in Lemma \ref{lem:settlement_CP}. We state the relation formally in Lemma \ref{lem:intensive_extensive_CQ}, but omit the proof. 
\begin{lemma}\label{lem:intensive_extensive_CQ}
Fix a set of honest players $\mathcal{I} \subseteq \mathcal{H}$ and parameters $\mu \in (0, 1)$, $T, k \in \mathbb{N}$. If the intensive chain quality property holds with parameters $\mu, s, k$ and $\mathcal{I}$ for all $s \leq T$, then the extensive chain quality property holds with parameters $\mu, T, k$ and $\mathcal{I}$. 
\end{lemma}

\subsection{Main Result}\label{sec:main_results}

\begin{definition}[$\epsilon$-honest majority]\label{def:eps_honest_maj}
    Consider a blockchain protocol where the leader election process has parameters $\alpha$ and $f$; and the communication network's typical delay is represented by a random variable $\Delta$. Let $G \sim \Geom (f)$ be a random variable that is independent of $\Delta$. Let $p \triangleq \alpha\, \mathbb{P}(\Delta < G).$    Suppose the system's parameters are such that $p > 0.5$. Let $\epsilon$ be such that $p = (1+\epsilon)/2$. We say that such a protocol has \emph{$\epsilon$-honest majority}.
\end{definition}

Note that for any $\alpha > 0.5$ and $\Delta$ such that $\mathbb{P}(\Delta < \infty) = 0$, one can choose $f > 0$ such that $p > 0.5$. 

Our main result states that the intensive safety and liveness properties hold with high probability, irrespective of the behavior of the adversary. 
\begin{theorem}[Main Result]\label{thm:main}
    Consider a blockchain protocol with $\epsilon$-honest majority. Then for any $\mathcal{I} \subseteq \mathcal{H}$, $s\in \mathbb{N}$ and $k \in \mathbb{N}$,
    \iftoggle{arxiv}
    {
    \begin{equation}  \label{eq:main_settlement_bnd}
        \mathbb{P}((s, k, \mathcal{I})\text{-settlement property is violated})
        \leq p_{\textsf{settlement}} + |\mathcal{I}| p_{\textsf{unheard}} 
    \end{equation}
    }
    {
    \begin{multline}  \label{eq:main_settlement_bnd}
        \mathbb{P}((s, k, \mathcal{I})\text{-settlement property is violated})
        \\
        \leq p_{\textsf{settlement}} + |\mathcal{I}| p_{\textsf{unheard}} 
    \end{multline}
    }
    where
\begin{align*}
   p_{\textsf{settlement}} &= \exp{(-kf \epsilon^3/12)} + 3\exp{(-kf\epsilon^2/32)} \\
   p_{\textsf{unheard}}  & = \left[\frac{2}{ 1-(1/2)^{\epsilon/2}}\right] \exp(-kf \epsilon/16)
\end{align*}
Further, for any $\mu < \epsilon$,
    \iftoggle{arxiv}
    {
    \begin{equation}  \label{eq:main_quality_bnd}
        \mathbb{P}((\mu, s, k, \mathcal{I})\text{-intensive chain quality property is violated}) \leq p_{\textsf{CQ}} + |\mathcal{I}| \Tilde{p}_{\textsf{unheard}} 
    \end{equation}
    }
    {
    \begin{multline}  \label{eq:main_quality_bnd}
        \mathbb{P}((\mu, s, k, \mathcal{I})\text{-intensive chain quality property is violated}) \\ \leq p_{\textsf{CQ}} + |\mathcal{I}| \Tilde{p}_{\textsf{unheard}} 
    \end{multline}
    }
    where
\begin{align*}
   p_{\textsf{CQ}} &= 4 \exp(-kf (\epsilon-\mu)^2 /48)   \\
   \tilde{p}_{\textsf{unheard}} & =   \left[\frac{2}{ 1-(1/2)^{\epsilon/2}}\right] \exp(-kf(\epsilon - \mu) /8)
\end{align*}
\end{theorem}

As a corollary, we get the following result about the extensive safety and liveness properties. This statement follows from Theorem \ref{thm:main}, Lemmas \ref{lem:settlement_CP} and \ref{lem:intensive_extensive_CQ}, and the union bound. We omit a formal proof. 
\begin{corollary}\label{thm:corollary}
    Consider a blockchain protocol with $\epsilon$-honest majority. Then for any $\mathcal{I} \subseteq \mathcal{H}$, $T\in \mathbb{N}$ and $k \in \mathbb{N}$,
    \iftoggle{arxiv}
    {
    \begin{equation*}
        \mathbb{P}((T, k, \mathcal{I})\text{-common prefix property is violated}) \leq T\,(p_{\textsf{settlement}} + |\mathcal{I}| p_{\textsf{unheard}})
    \end{equation*}
    }
    {
    \begin{multline*}
        \mathbb{P}((T, k, \mathcal{I})\text{-common prefix property is violated}) \\ \leq T\,(p_{\textsf{settlement}} + |\mathcal{I}| p_{\textsf{unheard}})
    \end{multline*}
    }
Further, for any $\mu < \epsilon$,
\iftoggle{arxiv}
    {
    \begin{equation*}
        \mathbb{P}((\mu, T, k, \mathcal{I})\text{-extensive chain quality property is violated})\leq T\,(p_{\textsf{CQ}} + |\mathcal{I}| \Tilde{p}_{\textsf{unheard}})
    \end{equation*}
    }
    {
    \begin{multline*}
        \mathbb{P}((\mu, T, k, \mathcal{I})\text{-extensive chain quality property is violated}) \\ \leq T\,(p_{\textsf{CQ}} + |\mathcal{I}| \Tilde{p}_{\textsf{unheard}})
    \end{multline*}
    }
\end{corollary}
The key difference between intensive and extensive properties can be seen by the guarantees on them. The probability of an intensive property, say the ($T, k, \mathcal{I}$)-settlement property, being violated is independent of $T$ (Theorem \ref{thm:main}). The probability of an extensive property, say the ($T, k, \mathcal{I}$)-common prefix property, being violated grows linearly with $T$ (Corollary \ref{thm:corollary}).

\begin{figure}[htbp]
    \centering
    \iftoggle{arxiv}
    {\includegraphics[width = 0.7\textwidth]{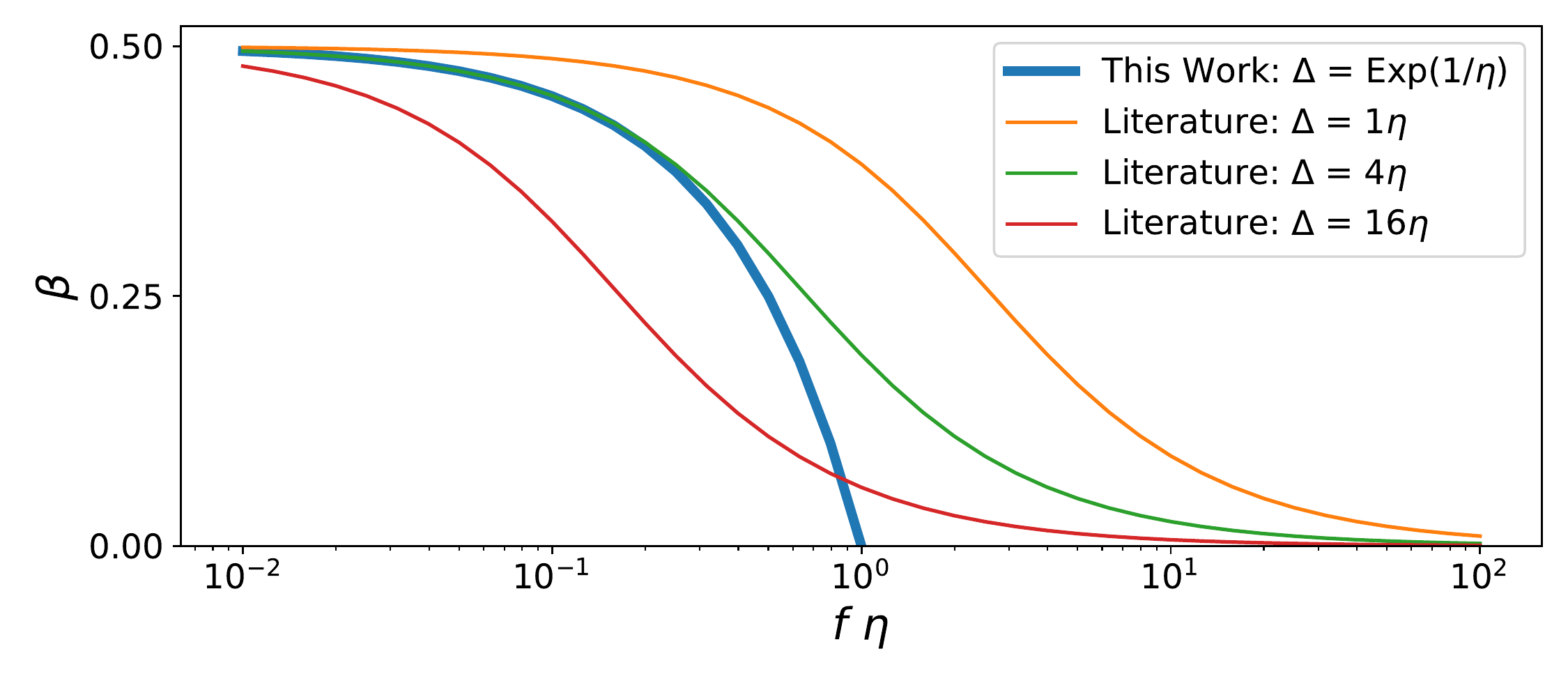}}
    {\includegraphics[width = \columnwidth]{formal-writeup/figures/beta.pdf}}
    \caption{Comparison of the security threshold guaranteed by this work for exponentially distributed delays with the tight threshold for deterministic delays  (i.e., \eqref{eq:adv_threshold}).}
    \label{fig:plot_exponential_delay}
\end{figure}

Theorem \ref{thm:main} and its corollary prove security properties under the $\epsilon$-honest majority assumption for some $\epsilon > 0.$ While the results are stated for discrete time, they imply corresponding results in continuous time by taking a limit, as described in \cite{gazi2020tight}. Also, in continuous time, the probability of more than one party mining at a time is zero.   Leaders are elected at times of a Poisson process of rate $f$, with a leader being the adversary with probability $\beta$ and honest with probability $1-\beta.$ The honest majority condition reduces to $(1-\beta)\mathbb{P}(\Delta < \mathsf{Exp}(f)) > \frac 1 2,$ where $\mathsf{Exp}(f)$ denotes an exponentially distributed random variable with rate parameter $f$ (mean $1/f$).  In case $\Delta$ has the $\mathsf{Exp}(1/\eta)$ distribution (with mean $\eta$), the honest majority condition becomes $\beta < \frac{1-\eta f} 2.$  Figure \ref{fig:plot_exponential_delay} displays the boundary of the security region we have established (i.e. $\beta = \frac{1-\eta f} 2$), and for comparison, the boundaries of the security region guaranteed for bounded delay by \eqref{eq:adv_threshold} for $\Delta\equiv \eta,$ $\Delta\equiv 4\eta,$ and $\Delta\equiv 16 \eta.$ Consider a delay distribution that is identical to $\mathsf{Exp}(1/\eta)$ from $[0, 4\eta]$, and concentrates the rest of the mass at $4 \eta$. Such a distribution can be stochastically dominated by both $\mathsf{Exp}(1/\eta)$, as well as the constant delay $4\eta$. The Figure shows that the adversarial tolerance guarantees provided by this work with $\mathsf{Exp}(1/\eta)$ delays are comparable to the best possible guarantees with constant delays $4 \eta$, for the range $f\eta < 0.2$.

\section{Definitions and Preliminary Results}\label{sec:definitions}

In this section, we define new terms pertaining to our model that are key to the proof of Theorem \ref{thm:main}. The two most important terms are $\CharString$ and $\Unheard$, defined in Sections \ref{sec:special_honest_slots} and \ref{sec:unheard} respectively. 

\subsection{Notation}\label{sec:notation}
All random processes in our model are discrete-time processes, indexed by $\mathbb{N}$, $\mathbb{Z}_{+}$ or $\mathbb{Z}$ (the relevant indexing will be specified when the process is defined). For a random process $\textsf{Process}$, the notation for the $i\textsuperscript{th}$ variable is $\textsf{Process}[i]$. The portion of the process from index $i_1$ to $i_2$, both inclusive, is denoted by $\textsf{Process}[i_1:i_2]$. If $i_2 < i_1$, this denotes an empty string. The process from index $i$ onward (including $i$) is denoted by $\textsf{Process}[i: \ ]$, and the process up to index $i$ (including $i$) is denoted by $\textsf{Process}[\ : i]$.

In our analysis, we often consider processes taking values in $\CharSymbols{}$. For such processes, define the following sets of time slots:
\begin{align*}
    \mathcal{N}_0(\textsf{Process}[i_1: i_2]) &\triangleq \{i \in \mathbb{N} : i_1 \leq i \leq i_2, \textsf{Process}[i] = 0\} \\
    \mathcal{N}_1(\textsf{Process}[i_1: i_2]) &\triangleq \{i \in \mathbb{N} : i_1 \leq i \leq i_2, \textsf{Process}[i] = 1\} \\
    \mathcal{N}(\textsf{Process}[i_1: i_2]) &\triangleq \{i \in \mathbb{N} : i_1 \leq i \leq i_2, \textsf{Process}[i] \neq \perp\}
\end{align*}
We denote the cardinality of these sets by using $N$ instead of $\mathcal{N}$. For example,
$N_0(\textsf{Process}[i_1 : i_2]) = |\mathcal{N}_0(\textsf{Process}[i_1 : i_2])|.$

\subsection{\textsf{LeaderString} and the compressed time scale}\label{sec:compressed_time_scale}
We start by defining a representation of the leader election process---$\LeaderString$---that we use in our analysis.
\begin{definition}[LeaderString]\label{def:leader_sequence}
\emph{$\LeaderString$} is a process taking values in $\CharSymbols$, defined as follows. For each $i \geq 1$,
\begin{equation}    \label{eq:leader_sequence_dist}
    \LeaderString[i] =  \left\{
    \begin{array}{lll}
    \perp & \text{if } N_i = 0,\, A_i = 0  &\mbox{(prob. } 1-f)  \\
    0 & \text{if } N_i = 1,\, A_i = 0 &\mbox{(prob. } \alpha f) \\
    1 & \text{if } N_i > 1 \text{ or } A_i = 1 &\mbox{(prob. } (1-\alpha)f )  \\
    \end{array}  \right.  
\end{equation}
\end{definition}
By the properties of the leader election process in Section \ref{sec:leader_election}, $\LeaderString$ is an i.i.d. process with the probabilities shown in \eqref{eq:leader_sequence_dist}.
We call a slot $i$ \textit{empty} if $\LeaderString[i] = \,\perp$ (and \textit{non-empty} otherwise). We call a slot $i$ \textit{uniquely honest} if $\LeaderString[i] = 0$. 

Let $(\LeaderString[i]:  i \leq 0)$ be a sequence of i.i.d. random variables with the same distribution as given in \eqref{eq:leader_sequence_dist}.  With this extension, the set of non-empty slots forms a \emph{stationary renewal process} with lifetime distribution $\Geom(f).$   Given the locations of all the renewal points, the labels at the renewal points are i.i.d. Bernoulli random variables with $\mathbb{P}(0) = \alpha$.
Let $1\leq T_1 < T_2 < \ldots$ denote the non-empty slots of $\LeaderString$ from slot $1$ onward. Similarly, let $0\geq T_0 > T_{-1} > T_{-2} > \ldots$ index the non-empty slots before or up to slot zero, going backwards in time. For any $j \neq 1$, $T_j - T_{j-1}$ has distribution $\Geom(f)$, while $T_1-T_0 = T_1 + (1-T_0) - 1,$  so that $T_1-T_0$ is the sum of two independent $\Geom(f)$ random variables minus one. In the terminology of renewal theory, $T_1 - T_0$ is the {\em sampled lifetime} sampled at time 0.

Suppose slot $T_j$ is uniquely honest, for some $j \in \mathbb{N}$. The leader of the slot, denoted by $h_j$, broadcasts a message to all other honest parties, each of which have independent delays. Let $\textsf{delay}(T_j \rightarrow h)$ denote the delay from the leader of $T_j$ to an honest party $h \in \mathcal{H}$. Strictly speaking, $\textsf{delay}(T_j \rightarrow h)$ has distribution $\Delta$ for all $h \neq h_j$, and is equal to $0$ for $h = h_j$. For the sake of homogeneity, however, we pretend that honest leaders send themselves a message that is subject to random delay. We, thus, extend the notation $\textsf{delay}(T_j \rightarrow h)$ to all slots $T_j, j \in \mathbb{Z}$ and assign independent delay random variables to them. Then $\{\textsf{delay}(T_j \rightarrow h): h \in \mathcal{H}, j \in \mathbb{Z}\}$ are i.i.d. delay random variables. 

The renewal points of $\LeaderString$ defines a new time scale: the clock ticks by one whenever a new non-empty slot occurs. Call this event-driven time scale the \textit{compressed time scale}. The following notation is used to define processes on the compressed time scale. For $s \geq 0$ and $j \geq 1$, let
\begin{equation}\label{eq:def_T_s_j}
    T^s_j \triangleq \min\{i: N(\LeaderString)[s+1:s+i] = j\}
\end{equation}
In other words, $T^s_j$ is the $j\textsuperscript{th}$ renewal point strictly after time $s.$ Clearly, $T^0_j = T_j$. Note that, for any $s \geq 0,$ $T^s_1$ and the random variables $\{T^s_j - T^s_{j-1}\}_{j \geq 2}$ are i.i.d. with distribution $\Geom(f).$ Given any process $\textsf{Process}$ on the original time scale, denote its time-shifted, compressed version relative to reference slot $s$ as $\textsf{CompressedProcess}_s$, defined by
\begin{align}
\textsf{CompressedProcess}_s[0] &\triangleq \textsf{Process}[s]   \nonumber \\ 
\label{eq:def_compressed_process_s}
    \textsf{CompressedProcess}_s[j] & \triangleq \textsf{Process}[s+T^s_j] \  \quad \text{for} \ j \geq 1
\end{align}
For example, $\textsf{CompressedLeaderString}_s [1 : \ ]$ is an i.i.d. $0$-$1$ valued process with probability of $0$ equal to $\alpha$. In other words, it is a Bernoulli process with parameter $1-\alpha$.

\subsection{Special Honest Slots and \textsf{CharString}}\label{sec:special_honest_slots}

We now introduce a new concept called \textit{special honest slots}. We also introduce our definition of \emph{the characteristic string}, denoted by $\CharString$.
Special honest slots are a subset of uniquely honest slots and play a role similar to {\em $\Delta$-isolated slots} in  \cite{david2018ouroboros}. Namely, the blocks mined in special honest slots must be at distinct heights in $\mathcal{F}_i$, irrespective of the actions of the adversary. The process $\CharString$ is defined such that it marks special honest slots with symbol $0$, other non-empty slots with symbol $1$, and empty slots with symbol $\perp$. Thus, defining $\CharString$ is equivalent to identifying special honest slots among uniquely honest slots.  We first define $\CharString$ in the one-time leader model and then in the i.i.d. leader model.

\subsubsection{\textsf{CharString} for one-time leader model}
The process $\CharString$ is a process indexed by $\mathbb{Z}$, taking values in $\CharSymbols{}$. We describe its construction, conditioned on the entire leader election process being known. Let $\CharString[i] = \perp$ for all $i$ such that $\LeaderString[i] = \perp$ and let $\CharString[i] = 1$ for all $i$ such that $\LeaderString = 1.$ It remains to select special honest slots among uniquely honest slots. We first do so for negative time, where special honest slots have no real interpretation. For $j \leq 0$, if $\LeaderString[T_j] = 0$, randomly set $\CharString[T_j]=0$ with probability $\mathbb{P}(\Delta<T_j-T_{j-1}|T_j-T_{j-1})$, and let $\CharString[T_j]=1$ otherwise. The aforementioned choices are conditionally independent across all $j\leq 0.$

For positive time, special honest slots are labeled sequentially as follows. For $j \geq 1$, let $T_{j^*}$ denote the last special honest slot at or before slot $T_{j-1}$. Define slot $T_j$ to be special honest if it is uniquely honest and $R_j < T_j - T_{j-1}$, where $R_j \triangleq \textsf{delay}(T_{j^*}\to h_j)$. Note that $h_j$ receives the message from the previous special honest slot at time $T_{j^*} +  R_j$, which, if $T_j$ is special honest, satisfies $T_{j^*} +  R_j < T_{j^*} +  T_j - T_{j-1} \leq T_j.$  Thus, the condition for $T_j$ to be a special honest slot is sufficient, but not necessary, for $h_j$ to have received the message from the previous special honest slot.

\subsubsection{Internal representation and refreshed residuals}

The definitions in this section are used in the following to define special honest slots for the i.i.d. leader model. Consider a probability mass function (pmf) $\textsf{f}$ on $\mathbb{Z}_+$, and let $X$ be a random variable with this pmf ($\mathbb{P}(X = i) = \textsf{f}[i]$). The {\em failure rate function} of the distribution, $\textsf{FailureRate}$, is defined by 
\[\textsf{FailureRate}[i] \triangleq \frac{\textsf{f}[i]}{\sum_{j \geq i} \textsf{f}[i]} = \frac{\mathbb{P}(X = i)}{\mathbb{P}(X \geq i)} \quad \text{for each} \ i \geq 0\]
with the convention that $\textsf{FailureRate}[i]=1$ if $\mathbb{P}(X \geq i)=0$.  

A random variable with pmf $\textsf{f}$ can be constructed as follows. Let $D=\min\{i\geq 0 : \textsf{U}[i]\leq \textsf{FailureRate}[i]\}$ where $\textsf{U}= (\textsf{U}[0], \textsf{U}[1], \ldots )$ be a sequence of independent random variables that are each uniformly distributed on the interval $[0,1]$. We call $(\textsf{FailureRate},\textsf{U})$ the {\em internal representation} of $D$. If $D_1$ and $D_2$ are random variables with independent internal representations, then $D_1$ and $D_2$ are independent as well.

Given $d\geq 0$, define the {\em refreshed residual} of $D$ at elapsed time $d$ by $\mathsf{refresh}_d(D)=\min\{i\geq 0: \textsf{U}[i+d]\leq \textsf{FailureRate}[i]\}.$
Although $\mathsf{refresh}_d(D)$ depends on the internal representation of $D$, the internal representation is suppressed in the notation.

\begin{lemma} \label{lem:refresh_property}
Let $D$ be a $\mathbb{Z}_+$-valued random variable with an internal representation and let $d\geq 0.$  The following hold.
\begin{description}
    \item (a) $\mathsf{refresh}_d(D) \stackrel{d.}{=} D$.
    \item (b) The random variable $\min \{d, D\}$ is independent of $\mathsf{refresh}_d(D)$.  More generally, if  $0 = d_0 < d_1 < \cdots < d_n$ then for each $j \in [n]$,  $\min \{d_j, \mathsf{refresh}_{d_{j-1}}(D)\}$ and $\mathsf{refresh}_{d_n}(D)$ are mutually independent.
    \item (c) If $D$ has a non-decreasing failure rate function, $D\leq d + \mathsf{refresh}_d(D).$
\end{description}
\end{lemma}
\begin{proof} Statement (a) follows from $\textsf{U} \stackrel{d.}{=} \textsf{U}[d: \ ].$  The first statement in (b) follows from the facts that $\min \{d, D\}$ is determined by $\textsf{U}[0 : d-1]$ and $\mathsf{refresh}_d(D)$ is determined by $\textsf{U}[d: \ ]$.  The generalization in (b) similarly follows: the indicated random variables are functions of disjoint subsets of $U$.  (c) is proved as follows.
\begin{align*}
D & \leq \min\{i\geq d: \textsf{U}[i] \leq \textsf{FailureRate}[i]\}\\
& =  d + \min\{i\geq 0: \textsf{U}[i+d] \leq \textsf{FailureRate}[i+d]\} \\
& \leq d + \min\{i\geq 0: \textsf{U}[i+d] \leq \textsf{FailureRate}[i]\} \\
& = d +  \mathsf{refresh}_d(D).
\end{align*}
\end{proof}

\subsubsection{\textsf{CharString} for i.i.d. leader model} \label{sec:SHS}
In this section we define $\CharString$ in the i.i.d. leader model. Without loss of generality, we assume all message delays have independent internal representations. The definition of $\CharString$ is the same as in the one-time leader model, except that the variables $R_j$ are defined differently. For each $j\geq 1,$ let $R_j \stackrel{\triangle}{=} \mathsf{refresh}_{T_{j-1} - T_{j^*}}(\mathsf{delay}(T_{j^*} \to h_j))$. Just as before, define $T_j$ to be a special honest slot (i.e. $\CharString[j]=0$) if $\LeaderString[j]=0$ and $ R_j < T_j - T_{j-1}.$ Note that $h_j$ receives the message from the previous special honest slot at time $T_{j^*} +  \mathsf{delay}(T_{j^*} \to h_j)$. If $T_j$ is special honest, then by Lemma \ref{lem:refresh_property}(c),
\iftoggle{arxiv}
{
\begin{equation*}
T_{j^*} +  \mathsf{delay}(T_{j^*} \to h_j)  \leq T_{j^*} +  (T_{j-1}-T_{j^*}) +  R_j < T_{j-1} + T_{j} - T_{j-1} = T_j.
\end{equation*}
}
{
\begin{multline*}
T_{j^*} +  \mathsf{delay}(T_{j^*} \to h_j)  \leq T_{j^*} +  (T_{j-1}-T_{j^*}) +  R_j \\ < T_{j-1} + T_{j} - T_{j-1} = T_j.
\end{multline*}
}
Thus, just as for the one-time leader model, the condition for $T_j$ to be a special honest slot is sufficient, but not necessary, for $h_j$ to have received the message from the previous special honest slot.

\subsubsection{The distribution of \textsf{CharString}}

The second lemma in this section characterizes the distribution of the random process $\CharString.$   Some preliminaries are given first. All results in this section hold for both the i.i.d. leader model and the one-time leader model.

For $j\geq 2$, define the following information set (i.e. $\sigma$-algebra generated by the set of random variables shown):
\begin{align*}
\mathsf{Info}_j = \sigma
\left\{
\begin{array}{c}
\mbox{leader election process from slot 1 up to slot } T_{j-1}   \\
\mathsf{CharString}[:T_{j-1}] , ~ h_j
\end{array}
\right\}
\end{align*}
Note that the leader election process specifies the identities of the leaders of each slot.
\begin{lemma}  \label{lem:CharStringProperties} 
For any $j\geq 2,$
$ \mathsf{Info}_{j},$   $R_j,$ and $T_j-T_{j-1}$ are mutually independent, $T_j- T_{j-1}$ has the $\Geom(f)$ probability distribution, and $R_j$ has the same distribution as $\Delta.$ 
\end{lemma} 
\begin{proof}
In the one-time leader model, the lemma is true by the construction of $\CharString.$  The lemma is true in the i.i.d. leader model by
 the construction of $\CharString$ and Lemma \ref{lem:refresh_property} (a) and (b). 
\end{proof}

The main result of this section is the following lemma.
\begin{lemma}  [Renewal structure of $\CharString$] \label{lem:renewal_prop_CharString}
The sequence of non-empty slots of $\CharString$ forms a stationary renewal process with lifetime distribution $\Geom(f).$  Conditioned on the renewal times $(T_j: j\in \mathbb{Z}),$ the labels $(\CharString[T_j]: j\in \mathbb{Z})$ are independent and for all $j\in \mathbb{Z},$
\begin{align} \label{eq:CharStringConditional}
\mathbb{P}(\CharString[T_j]=0|T_{j''}:j''\in \mathbb{Z}) =
\alpha \mathbb{P}(\Delta < T_j - T_{j-1}|T_j - T_{j-1})
\end{align}
\end{lemma}

\begin{proof}
The first sentence is true because $\CharString$ has the same set of non-empty slots as $\LeaderString$. Equation \eqref{eq:CharStringConditional} is true by construction for $j\leq 1$. Consider the following statement for $j\geq 1:$

${\mathcal{S}}_j:$  
The sequence of non-empty slots up to time $T_j$,  $(T_{j''}: j'' \leq j),$ forms a stationary renewal process with lifetime distribution $\Geom(f)$  and conditioned on such process,  the labels $(\CharString[T_{j'}]: j'\leq j)$ are conditionally independent, and for any $j' \leq j$,
\iftoggle{arxiv}
{
\begin{equation*}
    \mathbb{P}(\CharString[T_{j'}]=0|T_{j''}:j''\leq j) = \alpha \mathbb{P}(\Delta < T_{j'} - T_{j'-1}|T_{j'} - T_{j'-1})
\end{equation*}
}
{
\begin{multline*}
    \mathbb{P}(\CharString[T_{j'}]=0|T_{j''}:j''\leq j) =
 \\ \alpha \mathbb{P}(\Delta < T_{j'} - T_{j'-1}|T_{j'} - T_{j'-1})
\end{multline*}
}
It is shown next that ${\mathcal{S}}_j$ is true for all $j\geq 1$ by induction on $j.$  The base case $j=1$ is true by the construction of $\CharString.$  Suppose ${\mathcal{S}}_{j-1}$ is true for some $j\geq 2.$
Note that $\textsf{Info}_j$ includes the information in $(T_{j''}: j'' \leq j-1),$ so
Lemma \ref{lem:CharStringProperties} shows that the next lifetime is independent of the previous ones and the probability the renewal point at the end of the next lifetime is labeled 0 depends on the lifetime in the appropriate way.   Therefore, ${\mathcal{S}}_j$ is true, completing the proof by induction that ${\mathcal{S}}_j$ holds for all $j\geq 1.$

The statement of Lemma \ref{lem:renewal_prop_CharString} pertains to the joint distribution of $((T_j, \CharString[T_j]) : j\in \mathbb{Z}),$ which by definition is a statement about any finite sub-collection of the variables involved.  For any finite sub-collection of the variables,  the truth of ${\mathcal{S}}_j$ for $j$ sufficiently large implies that the finite sub-collection of variables have the joint distribution specified by the lemma,  completing the proof of the lemma.
\end{proof}

Properties of $\CompressedCharString$ follow as a corollary of Lemma \ref{lem:renewal_prop_CharString}.
\begin{lemma}   \label{lem:CCS}
For any $s\geq 0$, 
\[\mathbb{P}(\mathsf{CompressedCharString}_s[1]=0 \vert \CharString[\ :s]) \geq p.\] Further, for any $s\geq 0$ and $j \geq 2$,
\[\mathbb{P}(\mathsf{CompressedCharString}_s[j]=0 \vert \mathsf{CharString}[\ :T^s_{j-1}]) = p.\]
Here, $p$ is the parameter defined in Definition \ref{def:eps_honest_maj}.
\end{lemma}
\begin{proof}
Lemma \ref{lem:renewal_prop_CharString} implies that the right-hand sides of the two statements to be proved are the same for all $s\geq 0,$  and the statements follow for $s=0$ by Lemma \ref{lem:renewal_prop_CharString} as well.
\end{proof}

\subsection{The \textsf{Unheard} process}\label{sec:unheard}  Every honest party suffers some delay in receiving messages from special honest slots, and therefore may not have heard of all the special honest broadcasts. To prove security guarantees for a certain party $h \in \mathcal{H}$, we consider the delays suffered by $h$ alone in receiving messages from the leaders of special honest slots. Likewise, if we wish to prove security guarantees for a group of honest parties $\mathcal{I} \subset \mathcal{H}$, then we consider the delays suffered by all the parties in $\mathcal{I}$, but not other honest parties. The only other relevant delays for the security guarantees for $\mathcal{I}$ are the delays among the honest leaders, and these are appropriately incorporated into the definition of special honest slots.

\begin{definition} (LatestHeard and Unheard) For an honest party $h$ and $i\geq 1,$ let $\LatestHeard_h[i]$ denote the special honest slot with greatest index that $h$ has heard by the end of slot $i$. That is,
\iftoggle{arxiv}
{
\begin{equation*}
    \LatestHeard_h[i] = \max\{i': 1\leq i' \leq i, \CharString[i']=0, i' + \mathsf{delay}[i'\to h] \leq i\},
\end{equation*}
}
{
\begin{multline*}
    \LatestHeard_h[i] = \\ \max\{i': 1\leq i' \leq i, \CharString[i']=0, i' + \mathsf{delay}[i'\to h] \leq i\},
\end{multline*}
}
with the convention that the maximum of an empty set is  $-\infty.$
Let $\Unheard_h[i]$ denote the number of special honest slots after the slot containing the most recent special honest broadcast heard by $h$ by slot $i$. That is, 
\iftoggle{arxiv}
{
\begin{equation*}
    \Unheard_h[i] =
    |i'': \max\{0,\LatestHeard_h[i]\} < i'' \leq i, \CharString[i'']=0\}|.
\end{equation*}
}
{
\begin{multline*}
    \Unheard_h[i] = \\
    |i'': \max\{0,\LatestHeard_h[i]\} < i'' \leq i, \CharString[i'']=0\}|.
\end{multline*}
}
Additionally, for any honest party $h,$ let $\CompressedUnheard_{h,s}$ be the compressed  process corresponding to process $\Unheard_h$ and reference slot $s$,
as in Section \ref{sec:compressed_time_scale}.
Finally, given a set of honest parties $\mathcal{I}$, let
\[\LatestHeard_\mathcal{I}[i]= \min_{h\in \mathcal{I}} \LatestHeard_h[i],\]
\[\Unheard_\mathcal{I}[i]  = \max_{h\in \mathcal{I}}\Unheard_h[i],\]
\[\CompressedUnheard_{\mathcal{I},s}[j]  = \max_{h\in \mathcal{I}}\CompressedUnheard_{h,s}[j].\]
\end{definition}

For example, $\Unheard_h[i] = 2$ means that by the end of slot $i$, $h$ had not heard the last two special honest slots occurring before or at $i$, and it either heard the third most recent special honest slot before slot $i$ or there were only two special honest slots during $[1:i].$ 
\begin{restatable}[]{lemma}{distUnheard}\label{lem:dist_unheard} 
Let $q =  \mathbb{P}(\Delta \leq \Geom(f)).$  Then the following statements hold: \\
(a)  For any $i\geq 1$, $\mathbb{P}(\Unheard_h[i] >  a) \leq (1-q)^a$ for all integers $a\geq 0.$  \\
(b)  For any $s\geq 1$, and $j\geq 1$, $\mathbb{P}(\CompressedUnheard_{h,s}[j] >  a) \leq (1-q)^a$  for all integers $a\geq 0.$
\end{restatable}
The proof of this lemma is given in Appendix \ref{app:lemma_unheard}. The following lemma is a consequence of Lemma \ref{lem:dist_unheard} and the union bound.

\begin{lemma}   \label{lem:Uheard_line_bnd}
For any $k' \in \mathbb{N}$, $B \geq 0$ and $c \geq 0$,
\iftoggle{arxiv}
{
\begin{equation}\label{eq:RHS_comp}
    \mathbb{P}(\textsf{CompressedUnheard}_{h,s}[j]  \geq B + c(j-k')\text{ for some } j \geq k')    
    \\ \leq  \left[\frac{1}{(1-q)(1-(1-q)^c)}\right] \exp(-Bq)  
\end{equation}
}
{
\begin{multline}\label{eq:RHS_comp}
    \mathbb{P}(\textsf{CompressedUnheard}_{h,s}[j]  \geq B + c(j-k')\text{ for some } j \geq k')    
    \\ \leq  \left[\frac{1}{(1-q)(1-(1-q)^c)}\right] \exp(-Bq)  
\end{multline}
}
\end{lemma}
\begin{proof}
Lemma \ref{lem:dist_unheard} implies
$
\mathbb{P}(\CompressedUnheard_{h,s}[j] \geq t) \leq (1 - q)^{t-1} \mbox{ for } t \in \mathbb{R}_+
$
Substituting $B + c(j-k')$ for $t$ and using the union bound by summing over the possible values of $j-k'$ yields
that the left-hand side of \eqref{eq:RHS_comp} is bounded from above by
\begin{align*}
\sum_{d=0}^{\infty}(1 - q)^{B+cd-1}  =  \left[\frac{1}{(1-q)(1-(1-q)^c)}\right] (1-q)^B.
\end{align*}
Since $(1-q)^B \leq \exp(-B q)$, the bound in equation \eqref{eq:RHS_comp} follows.
\end{proof}

\section{Lemmas on Deterministic Properties}\label{sec:deterministic_lemmas}

In this section, we deduce some necessary conditions for violations of settlement and chain quality. The main tool is the notion of a {\em fork}, which describes some constraints on the possible blocktrees in an execution. We then define {\em reach} and {\em margin}, which are functions of a characteristic string and its associated fork. These metrics, first introduced in the Ouroboros line of work \cite{kiayias2017ouroboros, david2018ouroboros, blum2020combinatorics}, prove useful for analyzing settlement and chain quality violations, and we adapt the Ouroboros definitions to our setting. We introduce the basic terminology used for analyzing forks in Sections \ref{sec:forks}-\ref{sec:viable_tines_balanced_forks}. Sections \ref{sec:settlement_balanced_forks}-\ref{sec:settlement_margin} then focus on the settlement and Section \ref{sec:intensive_chain_quality} focuses on chain quality.

\subsection{Forks}\label{sec:forks}

Recall from Section \ref{sec:preliminaries} that $\mathcal{F}_i$ is a labeled, directed tree, representing the set of all blocks produced until the end of slot $i$. 
$\mathcal{F}_i$ depends on two factors: the adversary's actions and the random components of the protocol beyond the adversary's control. The characteristic string separates these two factors, by capturing all components beyond the adversary's control. The characteristic string, thus, imposes constraints on the possible $\mathcal{F}_i$ that the adversary can construct. These constraints are aptly described by the notion of a fork, which we define next.
\begin{definition}[Fork]\label{def:fork}
Let $w \in \CharSymbols{}^*$ be a finite string. A \emph{fork with respect to $w$} is a directed, rooted tree $F = (V,E)$ with a labeling $\ell:V \rightarrow \{0\} \cup \mathcal{N}(w)$ that satisfies the following properties.
\begin{itemize}
    \item each edge of $F$ is directed away from the root
    \item the root $r \in V$ is given the label $\ell(r) = 0$ 
    \item the labels along any directed path are strictly increasing
    \item each index $s \in \mathcal{N}_0(w)$ is the label of exactly one vertex of $F$
    \item the function $\mathbf{d} : \mathcal{N}_0(w) \rightarrow \mathbb{N}$, defined so that $\mathbf{d}(s)$ is the depth in $F$ of the unique vertex $v$ for which $\ell(v) = s$, satisfies the following monotonicity property: if $s_1 < s_2$, then $\mathbf{d}(s_1) < \mathbf{d}(s_2)$
\end{itemize}
We use the notation $F \vdash w$ if $F$ is a fork with respect to $w$.
\end{definition}
We now show that in any execution, irrespective of the adversary's actions and the instantiations of the random components, $\mathcal{F}_i$ is a fork with respect to $\CharString[1:i]$ (i.e., it satisfies the five properties listed in Definition \ref{def:fork}). The first three properties follow from the basic properties of blockchains described in Section \ref{sec:preliminaries}.  The fourth property is immediate given that special honest slots are a subset of uniquely honest slots, and that every honest leader proposes exactly one block when it is chosen as a leader. The last property is implied by the fillowing two facts.  First, every honest leader builds a chain that is strictly longer than any of the chains it has heard previously. Second, every special honest slot's leader has heard of the previous special honest slot's broadcast in a previous slot (see Section \ref{sec:special_honest_slots}). 

All honestly held chains, $\mathcal{C}^h_s$, $s \leq i$, are considered to be \textit{tines} in $\mathcal{F}_i$, where tines are defined as follows:

\begin{definition}[Tine]\label{def:tine}
Let $w \in \CharSymbols{}^*$ be a finite string. Let $F \vdash w$ be a fork. A \emph{tine} $t$ of $F$ is a directed path starting from the root. This is denoted by $t \in F$. For any tine $t$ define \emph{$\text{length}(t)$} to be the number of edges in the path, and for any vertex $v$ define its depth to be the length of the unique tine that ends at $v$. also define $\ell(t)$ to be the label of the vertex at the end of $t$.
\end{definition}
In Section \ref{sec:viable_tines_balanced_forks}, we further characterize honestly held chains by defining \emph{viable tines}.

We end this section with an important note about terminology. If $w \in \CharSymbols{}^*$ is a finite string and $F=(V,E)$ is a fork with respect to $w$, we call a slot $i$ an {\em adversarial slot} if $w[i]=1$ and we call a vertex in $V$ an \emph{adversarial block} if its label is an adversarial slot. In particular, consider $\CharString$. We treat a slot $i$ with $\CharString[i] = 1$ as adversarial, \emph{even if} $\mathsf{LeaderString}[i] = 0$. In other words, we treat uniquely honest slots that are not special honest as adversarial.

\subsection{Reach and Margin}\label{sec:reach_margin}

In this subsection, we define the terms reach and margin. These were previously described in earlier works (\cite{kiayias2017ouroboros, david2018ouroboros, blum2020combinatorics}). For a single point of comparison, we refer to \cite{blum2020combinatorics}. Our definitions are different from those in \cite{blum2020combinatorics} in two minor respects. First, our definitions are with respect to characteristic strings in $\CharSymbols{}^*$, instead of $\{0, 1\}^*$ as in \cite{blum2020combinatorics}.  Second, in the following definition, $t_1 \nsim_s t_2$ or being $s$-disjoint means the tines do not share any nodes with label greater than or equal to $s$, whereas in \cite{blum2020combinatorics} it means the times do not share any nodes with label (strictly) greater than $s.$   The version we use is more natural for considering violations of the $s,k$ settlement property.

In what follows, $i \in \mathbb{N}$ and  $w \in \CharSymbols{}^i$ are arbitrary.

\begin{definition}[The $\sim$ relation]\label{def:disjoint_tines}
Let $F \vdash w$. For two tines $t_1$ and $t_2$ of $F$, write $t_1 \sim t_2$ if $t_1$ and $t_2$ share an edge; otherwise write $t_1 \nsim t_2$ and refer to them as \emph{disjoint tines}. For any $s \leq i$,  write $t_1 \sim_s t_2$ if $t_1$ and $t_2$ share a node with a label greater than or equal to $s$; otherwise,  write $t_1 \nsim_s t_2$ and call such tines \emph{$s$-disjoint}.
\end{definition}

\begin{definition}[Closed fork]\label{def:closed_fork}
A fork $F \vdash w$ is closed if every leaf in $F$ is special honest. In other words, every leaf in $F$ has a label from the set $\mathcal{N}_0(w)$. 
\end{definition}

\begin{definition}[Closure of a fork]\label{def:closure}
Given a fork $F \vdash w$, the closure of $F$, $\overline{F} \vdash w$ is a closed fork obtained from $F$ by trimming all trailing adversarial blocks from all tines of $F$.
\end{definition}

\begin{definition}[Gap, Reserve, Reach]\label{def:gap_reserve_reach}
For a closed fork $F \vdash w$ and its unique longest tine $\hat{t}$,  define the \emph{gap of a tine $t\in F$} by 
 $ \text{gap}(t) \triangleq \text{length}(\hat{t}) -  \text{length}(t).$
Define the \emph{reserve of $t$}, denoted $\text{reserve}(t)$, to be the number of adversarial indices in $w$ that appear after the terminating vertex of $t$. In other words, if $v$ is the last vertex of $t$, then
$\text{reserve}(t) \triangleq \left \vert \{ i > \ell(v)\,|\, w[i] = 1\}\right \vert.$
These quantities are used to define the \emph{reach of a tine $t$}: 
$\text{reach}(t) \triangleq \text{reserve}(t) - \text{gap}(t).$
\end{definition}

For the intuition behind these definitions, we refer the reader to \cite{kiayias2017ouroboros, blum2020combinatorics}.

\begin{definition}[Reach of a fork or string]\label{def:max_reach}
For a closed fork $F \vdash w$, define $\Reach(F,w)$ to be the largest reach attained by any tine of $F$ (i.e., $\Reach(F,w) \triangleq \max_{t \in F} \text{reach}(t)$). 
We overload this notation to denote the maximum reach over all closed forks with respect to a finite-length characteristic string $w$:
\[\Reach(w) \triangleq \max_{F \vdash w,\, F \text{closed}} \Reach(F,w).\]
\end{definition}
Note that $\Reach(F,w)$ is non-negative, because the longest tine of any fork always has non-negative reach.

\begin{definition}[Margin of a fork or string]\label{def:margin}
For a closed fork $F \vdash w$ and $s < i$, define the
margin of $(F,w)$ relative to $s$ by:
\emph{\[\Margin_s(F,w) \triangleq \max_{t_1 \nsim_s t_2} \min\{\textnormal{reach}(t_1), \textnormal{reach}(t_2)\}.\]}
Once again, we overload notation to denote the relative margin of a string.
\emph{\[\Margin_s(w) \triangleq \max_{F \vdash w,\,F \text{closed}} \Margin_s(F,w).\]}
\end{definition}

For an infinite string $w \in \CharSymbols{}^\mathbb{N}$, $\Reach$ and $\Margin_s$ obey the recursive formulae we state below in equations (\ref{eq:reach_recursive}) and (\ref{eq:margin_recursive}). These are similar to those in Lemmas 2 and 3 in \cite{blum2020combinatorics}, with minor differences accounting for the two factors mentioned at the beginning of this section. The inclusion of $\perp$'s is inconsequential, as we show here. Suppose, for some $i$, $w[i] = \perp$. Then $F \vdash w[1:i-1] $ if and only if $ F \vdash w[1:i]$. It follows from Definitions \ref{def:gap_reserve_reach}, \ref{def:max_reach} and \ref{def:margin} that $\Reach(w[1:i]) = \Reach(w[1:i-1])$ and $\Margin_s(w[1:i]) = \Margin_s(w[1:i-1])$. For the complete proof of the following recursions, we refer the reader to \cite{kiayias2017ouroboros, blum2020combinatorics}. For the sake of defining the recursions, we define these quantities for an empty string as well. 

Let $\Reach(w[1:0]) = 0$, and for $i \geq 1$,
\begin{equation}\label{eq:reach_recursive}
    \Reach(w[1:i]) = \begin{cases}
    \Reach(w[1:i-1]) &\text{if } w[i] = \,\perp \\
    \Reach(w[1:i-1]) + 1 &\text{if } w[i] = 1 \\
    \textsf{(Reach}(w[1:i-1])-1 )_+ &\text{if } w[i] = 0 
    \end{cases}
\end{equation}

Let $s \in \mathbb{N}$ and $i \in \mathbb{Z}_+$. Then $\Margin_s(w[1:i]) = \Reach(w[1:i]) \mbox{ for } i < s$, and for $i \geq s$,
\iftoggle{arxiv}
{
\begin{equation}\label{eq:margin_recursive}
    \Margin_s(w[1:i]) =  \begin{cases}
    \Margin_s(w[1:i-1]) &\text{if } w[i] = \,\perp \\
    \Margin_s(w[1:i-1]) + 1 &\text{if } w[i] = 1 \\
    \Margin_s(w[1:i-1]) &\text{if } w[i] = 0 \text{ and }\\ &   \Reach(w[1:i-1])  >  \Margin_s(w[1:i-1]) = 0 \\
    \Margin_s(w[1:i-1]) - 1 &\text{if } w[i] = 0 \text{ and }\\ &   \{(\Reach(w[1:i-1]) = 0 \text{ or } \Margin_s(w[1:i-1]) \neq 0)\}
    \end{cases}
\end{equation}
}
{
\begin{multline}\label{eq:margin_recursive}
    \Margin_s(w[1:i]) = \\ \begin{cases}
    \Margin_s(w[1:i-1]) &\text{if } w[i] = \,\perp \\
    \Margin_s(w[1:i-1]) + 1 &\text{if } w[i] = 1 \\
    \Margin_s(w[1:i-1]) &\text{if } w[i] = 0 \text{ and }\\ &   \Reach(w[1:i-1]) \\ & >  \Margin_s(w[1:i-1]) = 0 \\
    \Margin_s(w[1:i-1]) - 1 &\text{if } w[i] = 0 \text{ and }\\ &   \{(\Reach(w[1:i-1]) = 0 \\ & \text{ or } \Margin_s(w[1:i-1]) \neq 0)\}
    \end{cases}
\end{multline}
}
The fact that $\Margin_s(w[1:i]) = \Reach(w[1:i])$ for $i < s$ is not explicitly shown in \cite{blum2020combinatorics}, so we prove it here. It suffices to
show that $\Margin_s(F) = \Reach(F) $ for $ F \vdash w[1:i], F \text{ closed},$ and $i < s$. The desired property holds because any tine in $F$ would not have blocks with a label greater than or equal to $s$, and is therefore $s$-disjoint with itself. The second largest reach among all $s$-disjoint pairs of tines, i.e., $\Margin_s(F)$, is therefore equal to the largest reach among all tines $\Reach(F)$.

\subsection{Viable Tines and Balanced Forks}\label{sec:viable_tines_balanced_forks}

We now introduce the terms \textit{viable tines} and \textit{balanced forks}, which are borrowed from \cite{kiayias2017ouroboros, blum2020combinatorics} but modified appropriately to suit our analysis.
\begin{definition}[Viable Tine]\label{def:viable_tine}
    Let $i \in \mathbb{N}$ and $w \in \CharSymbols{}^i$ be given. Let $F \vdash w$ be a fork and let $t$ be a tine of $F$. Say that $t$ is \emph{viable} if for all \emph{$s \in \mathcal{N}_0(w), \ \mathbf{d}(s) \leq \text{length}(t)$}. 
    
    Similarly, $t$ is \emph{$l$-viable} if for all \emph{$s \in \mathcal{N}_0(w[1 : l]), \ \mathbf{d}(s) \leq \text{length}(t)$}.
\end{definition}
Note that viability of a tine $t$ is defined in the context of a fixed fork $F$ and characteristic string $w$ with  $F \vdash w $ ($i$ is implicit; it's the length of $w$). When specializing to $\mathcal{F}_i \vdash \CharString[1:i]$, $l$-viable tines have the following interpretation. For an honest party $h$, let $l = \LatestHeard_h[i]$. Then $\mathcal{C}^h_i$ is an $l$-viable tine in $\mathcal{F}_i$.
We note some useful facts concerning viable tines. These facts are used in the proofs of the subsequent lemmas.
\begin{itemize}
    \item Given $i \in \mathbb{N}, w \in \CharSymbols{}^i$ and $F \vdash w$, a viable tine in $F$ is equivalent to an $i$-viable tine. If a tine is $l_1$-viable, it is also $l_2$-viable for every $l_2 < l_1$.
    \item If $t_1$ is an $l$-viable tine in $F$, and $t_2 \in F$ is a tine that is at least as long as $t_1$, then $t_2$ is also an $l$-viable tine.
    \item If $t \in F$ is at least as long as the longest tine in $\Bar{F}$, $t$ is viable in $F$.
\end{itemize}

\begin{definition}[Balanced Forks]\label{def:balanced_forks}
    Let $i \in \mathbb{N}, w \in \CharSymbols{}^i$, and $s \in \mathbb{N}$ such that $s \leq i$. A fork \emph{$F \vdash w$} is \emph{$s$-balanced} if it contains two tines $t_1, t_2$ s.t. both tines are viable and $t_1 \nsim_s t_2$. Similarly, $F$ is \emph{($s, l$)-balanced} if it contains two tines $t_1, t_2$ s.t. both tines are $l$-viable and $t_1 \nsim_s t_2$.
\end{definition}

In principle, we could allow for $s > i$ in the above definition. However, all forks $F \vdash w$ are $s$-balanced if $s > i$. This is because the longest tine in a fork is always viable, and it is $s$-disjoint with itself if $s > i$. Similarly, for any $l < s$, any fork is ($s, l$)-balanced. For any $l$, there is always an $l$-viable tine composed of blocks with labels $\leq l$ (the longest tine ending at a vertex with label in $\mathcal{N}_0(w[1:l])$. Such a tine is $s$-disjoint with itself.

Next, we introduce the notion of fork prefixes as they appear frequently in our proofs.

\begin{definition}[Fork Prefixes]\label{def:fork_prefixes}
    Let $i \in \mathbb{N}, w \in \CharSymbols{}^i$ and $i' \in \mathbb{N}$ such that $i' \leq i$ be given. For two forks $F \vdash w$, $F' \vdash w[1:i']$, say that $F'$ is a prefix of $F$ if $F'$ is a consistently labeled sub-graph of $F$. This is written as $F' \sqsubseteq F$.
\end{definition}
For every tine $t \in F$, there is a unique tine $t' \in F'$ with the vertices of $t'$ being the vertices of $t$ that are in $F'.$  Note that $\mathcal{F}_{i'} \sqsubseteq \mathcal{F}_{i}$ for any $i' < i$. In addition, for any $w \in \CharSymbols{}^*$ and any $F \vdash w$, $\Bar{F} \sqsubseteq F$. If $F'$ is a prefix of $F$,  say $F$ is a suffix of $F'$.

The notion of disjoint tines carries across forks that are prefixes of each other. Suppose $i \in \mathbb{N}, w \in \CharSymbols{}^i$ and $s \in \mathbb{N}$ such that $s \leq i$ are given.
Let $F \vdash w$ be a fork containing two tines $t_1, t_2$ such that $t_1 \nsim_s t_2$. For some $i' \leq i$, let $F' \vdash w[1:i']$ be a prefix of $F$, and let $t'_1$, $t'_2$ be tines corresponding to $t_1$ and $t_2$ respectively. Then $t'_1 \nsim_s t'_2$. A slightly technical point to note is that this statement holds irrespective of whether $i' \geq s$ or $i < s$; in the latter case, it is trivial as any tine $t$ such that $\ell(t) < s$ satisfies $t \nsim_s t$.

\subsection{Settlement and Balanced Forks}   \label{sec:settlement_balanced_forks}
We first introduce some terminology to reason about events concerning the settlement property for a given execution.   Given $s, k \geq 1$, and a subset $\mathcal{I}$ of honest parties, we define the event:
\[\mathcal{E}_{\text{settlement}} \triangleq \{\forall \ h_1, h_2 \in \mathcal{I}, \ \forall i_1, i_2 \geq s+k, \mathcal{C}^{h_1}_{i_1}[1:s] = \mathcal{C}^{h_2}_{i_2}[1:s]\}\] 
and, for $i\geq 1$, we define the event:
\iftoggle{arxiv}
{
\begin{equation}
    \mathcal{E}_{i\text{-settlement}} \triangleq \{\forall \ h_1, h_2 \in \mathcal{I}, \ \mathcal{C}^{h_1}_{i}[1:s] = \mathcal{C}^{h_2}_{i}[1:s]\}  \cap \{\forall \ h \in \mathcal{I}, \mathcal{C}^{h}_{i}[1:s] = \mathcal{C}^{h}_{i+1}[1:s]\} \label{eq:def_i_settlement}
\end{equation}
}
{
\begin{multline}
    \mathcal{E}_{i\text{-settlement}} \triangleq \{\forall \ h_1, h_2 \in \mathcal{I}, \ \mathcal{C}^{h_1}_{i}[1:s] = \mathcal{C}^{h_2}_{i}[1:s]\}  \\ \cap \{\forall \ h \in \mathcal{I}, \mathcal{C}^{h}_{i}[1:s] = \mathcal{C}^{h}_{i+1}[1:s]\} \label{eq:def_i_settlement}
\end{multline}
}
From these definitions, we deduce that
\iftoggle{arxiv}
{
\begin{equation}\label{eq:def_i_settlement_complement}
    \mathcal{E}^c_{i\text{-settlement}} = \{\exists \ h_1, h_2 \in \mathcal{I} \text{ such that } \mathcal{C}^{h_1}_{i}[1:s] \neq \mathcal{C}^{h_2}_{i}[1:s]\} \cup \{\exists \ h \in \mathcal{I} \text{ such that } \mathcal{C}^{h}_{i}[1:s] \neq \mathcal{C}^{h}_{i+1}[1:s]\} 
\end{equation}
}
{
\begin{multline}\label{eq:def_i_settlement_complement}
    \mathcal{E}^c_{i\text{-settlement}} = \{\exists \ h_1, h_2 \in \mathcal{I} \text{ such that } \mathcal{C}^{h_1}_{i}[1:s] \neq \mathcal{C}^{h_2}_{i}[1:s]\} 
    \\ \cup \{\exists \ h \in \mathcal{I} \text{ such that } \mathcal{C}^{h}_{i}[1:s] \neq \mathcal{C}^{h}_{i+1}[1:s]\} 
\end{multline}
}
Say that the settlement property with parameters $s, k, \mathcal{I}$ is \emph{violated at slot $i$} if $\mathcal{E}^c_{i\text{-settlement}}$ occurs. In words, this means that there exist two different honest parties who hold chains at slot $i$ that do not agree on slots up to $s$, or there exists an honest party whose chain at slot $i+1$ does not agree with its chain at slot $i$ on slots up to $s$. Suppose the $i$-settlement property is \textit{not} violated for any slot $i$ such that  $i \geq s + k$. Then all honestly held chains (among those in $\mathcal{I}$) agree up to slot $s$, from slot $s+k$ onward. This can be argued by induction.  Therefore,
$    \mathcal{E}_{\text{settlement}} = \bigcap_{i \geq s + k} \mathcal{E}_{i\text{-settlement}}$
or, equivalently,
\begin{equation}\label{eq:settlement_splitting}
\mathcal{E}^c_{\text{settlement}} = \bigcup_{i \geq s + k} \mathcal{E}^c_{i\text{-settlement}}
\end{equation}

We now state a relation between balanced forks and settlement violation.  
\begin{restatable}[Settlement Violation and Balanced Forks]{lemma}{settleBalanceFork}\label{lem:settle_balance_fork}
Suppose, in an execution, the settlement property for some $(s, k, \mathcal{I})$ is violated at slot $i$ (i.e., \emph{$\mathcal{E}^c_{i\text{-settlement}}$} occurs). Let \emph{$l = \LatestHeard_\mathcal{I}[i]$}. Then \emph{$ F \vdash \CharString[1:i]$} 
for some  $(s, l)$-balanced fork $F.$
\end{restatable}
The proof of this lemma is given in Appendix \ref{app:deterministic_lemmas}.

\subsection{Balanced Forks and Margin}
Lemma \ref{lem:settle_balance_fork} shows that settlement violations imply the existence of a balanced fork with respect to the characteristic string. We now derive an implication about the characteristic string alone. Towards this end, we first recall a lemma from \cite{blum2020combinatorics}. 
\begin{restatable}[from \cite{blum2020combinatorics}]{lemma}{balanceMargin}\label{lem:balanced_fork_mu}
    Let $i \in \mathbb{N}$, $w \in \CharSymbols{}^i$ and $s \in \mathbb{N}$ such that $s \leq i$. There exists an $s$-balanced fork $F \vdash w$ if and only if $\Margin_s(w) \geq 0$.
\end{restatable}
For completeness, we provide the proof in Appendix \ref{app:deterministic_lemmas}. The above lemma provides a characterization for the existence of $s$-balanced forks $F \vdash w$. However, we are interested in characterizing a more general form of balanced forks, i.e., $(s, l)$-balanced forks. We show that every $(s, l)$-balanced fork can be mapped to an $s$-balanced fork and vice-versa (Lemma \ref{lem:balanced_fork_equivalence}). First define a useful transformation on strings in $\CharSymbols{}^*$ that will be used in this lemma.
\begin{definition}[$O_l(w)$]\label{def:observer_char_string}
    Let $i \in \mathbb{N}$, $w \in \CharSymbols{}^i$, and $l \in 
    \mathbb{Z}_+$ such that $l \leq i$. Then $O_l(w) \in \{0, 1, \perp\}^i$ is a string obtained from $w$ by replacing each 0 in $w[l+1:i]$ by $\perp$.
\end{definition}
$O_l(\cdot)$ is a map from $\CharSymbols{}^* \rightarrow \CharSymbols{}^*$. It has the following interpretation. For any $i \in \mathbb{N}$, let $l = \LatestHeard_h[i]$. Then $O_l(\CharString[1:i])$ is effectively the characteristic string observed by the honest party $h$, assuming the adversary delays all messages maximally. Since $h$ has not heard the broadcasts from the special honest slots after $l$, those slots are seen as empty slots by $h$. Note that this interpretation works only by assuming a certain adversarial action; the adversary may choose to reveal blocks from special honest slots in $[l+1:i]$ if it so wishes. 
The notion of fork prefixes can be extended naturally to forks $F \vdash w$, $F' \vdash w'$, where $w' = O_l(w)$. Given $i \in \mathbb{N}, w \in \CharSymbols{}^i, F \vdash w$ and $l \leq i$, drop all blocks with labels in $\mathcal{N}_0(w[l+1:i])$ and their descendants to obtain $F'$. It can be verified that such an $F'$ satisfies the rules of a fork with respect to $w'$. Clearly, $F'$ is a sub-tree of $F$ and we therefore say $F' \sqsubseteq F$.

\begin{restatable}{lemma}{equivalence}\label{lem:balanced_fork_equivalence}
Let $i \in \mathbb{N}$, $w \in \CharSymbols{}^i$, $s \in \mathbb{N}$ such that $s \leq i$, and $l \in \mathbb{Z}_+$ such that  $l \leq i$. Let $w' = O_l(w)$. There exists an $(s, l)$-balanced fork $F \vdash w$ if and only if there exists an $s$-balanced fork $F' \vdash w'$.
\end{restatable}

The proof of this lemma is given in Appendix \ref{app:deterministic_lemmas}. Combining Lemma \ref{lem:balanced_fork_equivalence} with Lemma \ref{lem:balanced_fork_mu} gives us the following corollary:

\begin{lemma}\label{lem:balanced_fork_mu_2}
    Let $i \in \mathbb{N}$, $w \in \CharSymbols{}^i$, $s \in \mathbb{N}, s \leq i$ and $l \in \mathbb{Z}_+, l \leq i$ be given. Then $\exists \, (s, l)$-balanced fork $F \vdash w$ if and only if \emph{$\Margin_s(O_l(w)) \geq 0$}.
\end{lemma}
\begin{proof}
    By Lemma \ref{lem:balanced_fork_equivalence}, $\exists \, (s, l)$-balanced fork $F \vdash w$ if and only if $\exists \, s$-balanced fork $F' \vdash O_l(w)$. By Lemma \ref{lem:balanced_fork_mu}, $\exists \, s'$-balanced fork $F \vdash O_l(w)$ if and only if $\Margin_s(O_l(w)) \geq 0$. Together, they imply $\exists \, (s, l)$-balanced fork $F \vdash w$ if and only if $\Margin_s(O_l(w)) \geq 0$.
\end{proof}

\subsection{Settlement and Margin}   \label{sec:settlement_margin}
Lemmas \ref{lem:settle_balance_fork} and \ref{lem:balanced_fork_mu_2} give the following necessary condition for violations of settlement:
\begin{lemma}[Settlement Violation]
   \label{lem:settlement_violation_margin}
If the settlement property with parameters $(s, k, \mathcal{I})$ is violated in an execution at slot $i$ (i.e., \emph{$\mathcal{E}^c_{i\text{-settlement}}$} occurs), then \emph{$\Margin_s(O_l(\CharString[1:i])) \geq 0$}, where \emph{$l = \LatestHeard_\mathcal{I}[i]$}.
\end{lemma}
\begin{proof}
By Lemma \ref{lem:settle_balance_fork}, if the settlement property with parameters $(s, k, \mathcal{I})$ is violated at slot $i$, then $\exists \, (s, l)$-balanced fork $F \vdash \CharString[1:i]$ such that $F$ is an ($s, l$)-balanced fork. By Lemma \ref{lem:balanced_fork_mu_2}, $\exists \, (s, l)$-balanced fork $F \vdash \CharString[1:i]$ if and only if $\Margin_s(O_l(\CharString[1:i])) \geq 0$. Thus, the statement of the lemma follows.
\end{proof}

The following lemma helps relate $\Margin_s(O_l(\CharString[1:i]))$ to $\Margin_s(\CharString[1:i])$:
\begin{lemma}\label{lem:reach_margin_unheard_bound}
    Let $w \in \CharSymbols{}^\mathbb{N}$ and $l, s \in \mathbb{N}$. Then, for any $i \geq l$,
    \begin{align*}
        \Reach(O_l(w[1:i])) &= \Reach(w[1:l]) + N_1(w[l+1:i])\\
       & \leq  \Reach(w[1:i]) + N_0(w[l+1:i]) \\
        \Margin_s(O_l(w[1:i])) &= \Margin_s(w[1:l]) + N_1(w[l+1:i])  \\
        &\leq \Margin_s(w[1:i]) + N_0(w[l+1:i])
    \end{align*}
\end{lemma}
\begin{proof}
We prove the result for $\Reach$ by induction; the result for $\Margin_s$ can be proven in an identical fashion. By re-arranging terms, the desired result can be stated in the following terms:
\begin{align*}
    \Reach(O_l(w[1:i])) &= \Reach(w[1:l]) + N_1(w[l+1:i]) \\ 
    \Reach(w[1:i]) &\geq \Reach(w[1:l]) + N_1(w[l+1:i])\iftoggle{arxiv}{}{\\
    &~~~}  - N_0(w[l+1:i])
\end{align*}
For the base case with $i = l$, we observe that $O_l(w[1:l]) = w[1:l]$, which implies $\Reach(O_l(w[1:l])) = \Reach(w[1:l])$, which is identical to the desired statement with $i = l$.
For any $i > l$, assume the desired statements hold for all $i' < i$. The key observation here is that for a fixed $l$, $\Reach(O_l(w[1:i]))$ satisfies \eqref{eq:reach_recursive}. This is because $O_l(w[1:i])$ is a string that is obtained by concatenating one additional symbol to $O_l(w[1:i-1])$.  
\begin{itemize}
    \item If $w[i] = \perp$, $\Reach(O_l(w[1:i])) = \Reach(O_l(w[1:i-1]))$ and $\Reach(w[1:i]) = \Reach(w[1:i-1])$. 
    \item If $w[i] = 1$, $\Reach(O_l(w[1:i])) = \Reach(O_l(w[1:i-1])) + 1$ and $\Reach(w[1:i]) = \Reach(w[1:i-1]) + 1$.
    \item If $w[i] = 0$, $\Reach(O_l(w[1:i])) = \Reach(O_l(w[1:i-1]))$ and $\Reach(w[1:i]) = \Reach(w[1:i-1])$ or $\Reach(w[1:i]) = \Reach(w[1:i-1]) - 1$. We can therefore say $\Reach(w[1:i]) \geq \Reach(w[1:i-1]) - 1$
\end{itemize}
(Crucially, these equations hold for $\Margin_s$ also, irrespective of the value of $s$.)

These equations can be summarized as:
\begin{align*}
    \Reach(O_l(w[1:i])) &= \Reach(O_l(w[1:i-1])) + N_1(w[i]) \\ \Reach((w[1:i])) &\geq \Reach(w[1:i-1]) + N_1(w[i])\iftoggle{arxiv}{}{\\
    &~~~}- N_0(w[1:i-1])
\end{align*}

By the induction hypothesis,
\begin{align*}
    \Reach(O_l(w[1:i-1])) &= \Reach(w[1:l]) + N_1(w[l+1:i-1]) \\ \Reach(w[1:i-1]) &\geq \Reach(w[1:l]) + N_1(w[l+1:i-1])\iftoggle{arxiv}{}{\\
    &~~~} - N_0(w[l+1:i-1])
\end{align*}
Combining these equations, we get the desired result.
\end{proof}

We now obtain the main lemma, which states a necessary condition for $\mathcal{A}$ to violate the settlement property.
\begin{lemma}[Settlement Violation-Necessary Condition]\label{lem:settlement_necessary}
    Suppose, in an execution, the settlement property with parameters $(s, k, \mathcal{I})$ is violated (i.e., \emph{$\mathcal{E}^c_{\text{settlement}}$} occurs). Then, for some $i \geq s+k$,
    \[\Margin_s(\CharString[1:i]) + \Unheard_\mathcal{I}[i] \geq 0.\]
\end{lemma}
\begin{proof}  
Suppose  $\mathcal{E}^c_{\text{settlement}}$ occurs.   Then, by  \eqref{eq:settlement_splitting}, there exists $i \geq s + k$ such that $\mathcal{E}^c_{i\text{-settlement}}$ occurs. 

By Lemma \ref{lem:settlement_violation_margin}, 
$\Margin_s(O_l(\CharString[1:i])) \geq 0,$  where  $l = \LatestHeard_\mathcal{I}[s+i].$

By Lemma \ref{lem:reach_margin_unheard_bound}, 
$\Margin_s(\CharString[1:i]) + N_0(\CharString[l+1:i]) \geq 0. $
Since $N_0(\CharString[l+1:i]) = \Unheard_\mathcal{I}[i],$ the result follows.
\end{proof}

It is interesting to contrast Lemma \ref{lem:settlement_necessary} with the corresponding statement in \cite{blum2020combinatorics}, which is given below in our notation:
\[\Margin_s(\CharString[1:i]) \geq 0 \text{ for some } i \geq s+k\]
Clearly, the delay model places a more stringent condition on $\Margin_s[i]$ for settlement to hold.

\subsection{Intensive Chain Quality}\label{sec:intensive_chain_quality}
In this section, we derive a necessary condition for violations of intensive chain quality. Recall the definition of intensive chain quality with parameters $s, k, f, \mu,$ and $\mathcal{I}$ from Definition \ref{def:chain_quality_intensive}: this property holds if any chain held by an honest party in $\mathcal{I}$ after slot $s +k$ has at least a fraction of $\mu$ honest blocks from the interval $\{s+1, \ldots, s+k\}$. We shall work with a stronger property, by replacing honest blocks by special honest blocks. So given $s,k \in \mathbb{N}$, $f, \mu > 0$ and a set of honest parties $\mathcal{I},$ let $\mathcal{E}_{\text{cq}},$ be the event that $\mathcal{C}^h_i[s+1:s+k]$ contains greater than $k \mu f \text{ special honest blocks}$  for all $i \geq s + k$ and all $h\in \mathcal{I}.$
$\mathcal{E}_{\text{cq}}$ implies intensive chain quality with the same parameters.
The main result of this section, Lemma \ref{lem:chain_quality_necessary}, gives a necessary condition for $\mathcal{E}_{\text{cq}}^c$, the event of an intensive chain quality violation, in terms of the characteristic string and related quantities.

Section \ref{sec:reach_margin} defines $\Reach$ as a mapping from strings to $\mathbb{Z}_+.$  If the string is $\CharString,$ let $\Reach$ denote the random process defined by $\Reach[s]=\Reach(\CharString[1:s]).$    

For any slot $i \in \mathbb{N}$, let $\mathcal{C}^*_{i}$ denote the chain broadcast by the leader of the last special honest slot at or before slot $i$. Since these chains must have strictly increasing lengths,
\[|\mathcal{C}^*_{i_2}| \geq |\mathcal{C}^*_{i_1}| + N_0(\CharString[i_1+1:i_2]), \forall i_1 \leq i_2.\]

The following lemma provides an upper bound on the length of any prefix of an honestly held chain.
\begin{lemma}  \label{lem:ICQ_lower}  For any $h \in \mathcal{H}$, for any $i,s \in \mathbb{N},$ 
\[|\mathcal{C}^h_{i}[1:s]| \leq |\mathcal{C}^*_{s}| + \Reach[s].\]
\end{lemma}
\begin{proof}
We first prove a more general result, stated for any (string, fork, tine) tuple. Let $i \in \mathbb{N}$, $w \in \CharSymbols{}^i$, a fork $F \vdash w$ and a tine $t \in F$ be given. Let $\Bar{F} \vdash w$ be the closure of $F$, and let $\Bar{t} \in \Bar{F}$ be the tine corresponding to $t$. Let $\hat{t}$ be the longest tine in $\Bar{F}$. Then,
\begin{align}
     \text{length}(t) &\leq \text{length}(\hat{t}) + \text{reach}(\Bar{t}) \leq \text{length}(\hat{t}) + \Reach(\Bar{F}) \nonumber  \\
& \leq \text{length}(\hat{t}) + \Reach(w)    \label{eq:length_bound}
\end{align}
The first inequality follows from Definition \ref{def:gap_reserve_reach}, while the second and third follow from Definition \ref{def:max_reach}.

To complete the proof of the lemma we explain why the claimed result is a special case of \eqref{eq:length_bound}. Let $w = \CharString[1:s]$, and let $F$ be the prefix of $\mathcal{F}_i$ obtained by dropping all blocks with label greater than $s$. Since $\mathcal{C}^h_{i}$ is a tine in $\mathcal{F}_i$, $\mathcal{C}^h_{i}[1:s]$ is a tine in $F$; denote it by $t$. Further, the longest tine in $\Bar{F}$ is the tine ending in the block labeled with the last special honest slot at or before $s$, which is precisely the tine $\mathcal{C}^*_s$. With this mapping, the desired inequality follows.
\end{proof}

We now define $\Advantage_s$ as follows:
\begin{align}
   & \textsf{Advantage}_s(\CharString[1:i])  \triangleq N_1(\CharString[s+1:i])   \iftoggle{arxiv}{}{\nonumber \\
    & ~~~~~ }- N_0(\CharString[s+1:i]) +  k f \mu + \Reach[s]  \label{eq:def_advantage}
\end{align}
$\Advantage_s$ is used in the lemma below.

\begin{lemma}[Intensive chain quality violation -- necessary condition] \label{lem:chain_quality_necessary}
Suppose intensive chain quality with parameters $s, k, f, \mu$ and  $\mathcal{I}$ is violated in an execution (i.e., $\mathcal{E}_{\text{cq}}^c$ occurs). Then, for some $i \geq s+k,$
\begin{align}  \label{eq:CQV}
\Advantage_s(\CharString[1:i]) + \Unheard_{\mathcal{I}}[i] \geq 0.
\end{align}
\end{lemma}

\begin{proof}
Consider an execution where $\mathcal{E}_{\text{cq}}^c$ occurs. There exist a slot $i \geq s + k$ and honest party $h \in \mathcal{I}$ such that
$N_0(\mathcal{C}^h_i[s+1:s+k])\leq k \mu f.$ 
First, consider the case that $\LatestHeard_h[i] < s$.  Then,
\begin{align*}
 N_0(\CharString[s+1:i]) \leq  \Unheard_h[i] \leq  \Unheard_{\mathcal{I}}[i].
\end{align*}
Combining this inequality with \eqref{eq:def_advantage} yields \eqref{eq:CQV}.

Next, consider the case that $\LatestHeard_h[i] \geq s.$ Let $i^*$ be the largest integer such that:
$s+k \leq i^* \leq i$ and there are no special honest slots in $\mathcal{C}^h_i[s+k+1:i^*].$  
We now show that
\begin{align} \label{eq:lowerICQ}
|\mathcal{C}^h_{i}[1:i^*]| \geq |\mathcal{C}^*_{s}| + N_0(\CharString[s+1:i^*]) - \Unheard_h[i^*].
\end{align}
The proof of \eqref{eq:lowerICQ} is divided into the cases $i^*<i$ and $i^*=i.$

($i^* < i$)  Suppose $i^* < i.$   Then $i^*+1$ is a special honest slot and the message sent by the leader $h'$ of slot $i^*+1$, $\mathcal{C}^{h'}_{i^*+1}$, is a prefix of $\mathcal{C}^h_i$.  Therefore, $\mathcal{C}^{h'}_{i^*+1}[1:i^*] = \mathcal{C}^h_i[1:i^*].$   Since, by the end of slot $i^*$, $h'$ received messages sent by the leaders of all special honest slots in $[1:i^*],$ it follows that
\begin{align*}
|\mathcal{C}^h_i[1:i^*]| = |\mathcal{C}^{h'}_{i^*+1}[1:i^*] |  \geq
|\mathcal{C}^*_{s}| + N_0(\CharString[s+1:i^*]).
\end{align*}
which implies \eqref{eq:lowerICQ}.

($i^* = i$)  Suppose $i^* = i.$  Let $l = \LatestHeard_h[i].$  Then $l\leq i$
and, by our prior assumption, $l \geq s.$   Therefore
\begin{align*}
|\mathcal{C}^h_{i}|& \geq |\mathcal{C}^*_{l}| \geq |\mathcal{C}^*_{s}| + N_0(\CharString[s+1:l]) \iftoggle{arxiv}{}{\\
&}= |\mathcal{C}^*_{s}| + N_0(\CharString[s+1:i]) - \Unheard_h[i],
\end{align*}
which, together with the fact $i=i^*$
(so $\mathcal{C}^h_{i} =\mathcal{C}^h_{i}[1:i^*]$), proves \eqref{eq:lowerICQ}.
This completes the proof of \eqref{eq:lowerICQ} in either case.

We now find an upper bound for $|\mathcal{C}^h_{i}[1:i^*]|.$
We know that:
\begin{itemize}
    \item $|\mathcal{C}^h_i[1:s]| \leq |\mathcal{C}^*_{s}| + \Reach[s]$, by Lemma \ref{lem:ICQ_lower}.
    \item $|\mathcal{C}^h_i[s + 1:s + k]| \leq k \mu f + N_1(\CharString[s + 1:s + k])$, because, by assumption, at most $ k \mu f$ blocks in $\mathcal{C}^h_i[s + 1:s + k]$ are from special honest slots; the rest must have labels in $\mathcal{N}_1(\CharString[s + 1:s + k])$.
    \item $|\mathcal{C}^h_i[s + k + 1:i^*]| \leq N_1(\CharString[s + k + 1:i^*])+1$, because none of the blocks in $\mathcal{C}^h_i[s + k + 1:i^*]$ are from special honest slots.
\end{itemize}
Together, we get 
\begin{align}
    |\mathcal{C}^h_{i}[1:i^*]| &= |\mathcal{C}^h_i[1:s]| + |\mathcal{C}^h_i[s + 1:s + k]| + |\mathcal{C}^h_i[s + k + 1:i^*]|  \nonumber \\
    &\leq |\mathcal{C}^*_{s}| + \Reach[s] + k f \mu + N_1(\CharString[s + 1:s + k])\iftoggle{arxiv}{}{ \nonumber \\ 
    &~~~~~} + N_1(\CharString[s + k + 1:i^*])  \nonumber \\
    &= |\mathcal{C}^*_{s}| + \Reach[s] + k f \mu + N_1(\CharString[s + 1:i^*])   \label{eq:upperICQ}
\end{align}
Combining \eqref{eq:def_advantage}, \eqref{eq:lowerICQ}, and \eqref{eq:upperICQ} yields \eqref{eq:CQV}. Thus, the lemma holds.
\end{proof}

\section{Proof Sketch of Theorem \ref{thm:main}}\label{sec:probabilistic_lemmas}
We provide a proof sketch in this section and defer the full proof to Appendix \ref{app:probabilistic_lemmas}. The proof of Theorem \ref{thm:main} relies primarily on the properties of $\CharString$ (Lemmas \ref{lem:renewal_prop_CharString} and \ref{lem:CCS}) and the bounds on $\Unheard$ (Lemmas \ref{lem:dist_unheard} and \ref{lem:Uheard_line_bnd}). 

As Lemmas \ref{lem:settlement_necessary} and \ref{lem:chain_quality_necessary} provide necessary conditions for violations of settlement and chain quality, bounding their probabilities is sufficient to prove security. In other words, it suffices to prove the following two statements:
\iftoggle{arxiv}
{
\begin{align*}
    \mathbb{P} \left( \Margin_s(\CharString[1:i]) + \Unheard_\mathcal{I}[i] \geq 0 \text{ for some } i \geq s+k \right)
    &\leq p_{\textsf{settlement}} + |\mathcal{I}|p_{\textsf{unheard}}  \\   \mathbb{P}\left(\Advantage_s(\CharString[1:i]) +  \Unheard_{\mathcal{I}}[i] \geq 0 \text{ for some } i \geq s+ k\right)
    &\leq p_{\textsf{CQ}} + |\mathcal{I}|  \tilde{p}_{\textsf{unheard}}
\end{align*}
}
{
\begin{align*}
    \mathbb{P} \left(
    \begin{array}{c}
    \Margin_s(\CharString[1:i]) + \Unheard_\mathcal{I}[i] \geq 0 \\
    \text{ for some } i\geq s+k 
    \end{array} \right) \nonumber \\
    \leq p_{\textsf{settlement}}
    + |\mathcal{I}|p_{\textsf{unheard}}  \\ 
    \mathbb{P}\left(  \begin{array}{c}
    \Advantage_s(\CharString[1:i]) +  \Unheard_{\mathcal{I}}[i] \geq 0 \\ \text{ for some } i \geq s+ k
    \end{array} \right)  \nonumber \\
    \leq p_{\textsf{CQ}} + |\mathcal{I}|  \tilde{p}_{\textsf{unheard}}
\end{align*}
}
In Appendix \ref{sec:time_reduction}, we derive events on the compressed time scale that are implied by the events on the left-hand sides of these two statements. Analyzing these new events is therefore sufficient to prove security. In Appendix \ref{sec:reach}, we use Lemma \ref{lem:renewal_prop_CharString} to show that $\Reach[s]$ is stochastically dominated by a geometric random variable.

By Lemma \ref{lem:CCS}, $\CompressedCharString_s$ is (nearly) a Bernoulli process. The difference between the number of adversarial blocks and special honest blocks as a function of time behaves, therefore, like a random walk with negative drift. In Appendix \ref{sec:random_walk}, we bound such a process from above by an affine function with negative slope.

The results of Appendices \ref{sec:reach} and \ref{sec:random_walk} translate into affine bounds for the compressed time scale analogues of $\Margin_s$ and $\Advantage_s$. In the case of $\Margin_s$, we extend a result of \cite{blum2020combinatorics}. In Appendices \ref{sec:margin} and \ref{sec:advantage}, we combine these affine bounds with Lemma \ref{lem:Uheard_line_bnd} to prove the desired statements on settlement and chain quality.
\iftoggle{arxiv}{}{\clearpage}

\bibliographystyle{alpha}
\bibliography{blockchain_papers}

\appendix

\section{Proof of Lemma \ref{lem:dist_unheard}}   \label{app:lemma_unheard}
\distUnheard*
\begin{proof}
Fix $i\geq 1.$ It is possible that $i$ itself is a special honest slot and $h$ has not heard it by slot $i$. In any case, $\Unheard_h[i] \leq 1 + N_0(\CharString[\LatestHeard[i]:i])$, i.e., $\Unheard_h[i]$ is less than or equal to one plus the number of consecutive special honest slots from strictly before slot $i$ that $h$ has not heard by slot $i$.

The non-empty slots of $\CharString$ form both a Bernoulli process with parameter $f$ and a renewal process. Let $D_1, D_2, \ldots $ denote the lifetimes of the renewal process going backwards from slot $i$.
Thus, $i-D_1 - \cdots - D_j$ is the $j\textsuperscript{th}$ non-empty slot of $\CharString$ (strictly) before $i$.
The random variables $D_i$ are independent with the $\Geom(f)$ distribution.
The last special honest slot before slot $i$ must be at least $D_1$ slots before slot $i$, so the probability $h$ has heard that special honest slot is at least $q$. In general, for $j\geq 1$, the $j\textsuperscript{th}$ from the last special honest slot before slot $i$ must be at least $D_j$ slots before slot $i.$  (Here $D_j$ is used as a lower bound on $D_1 + \cdots + D_j.$)  Thus, no matter which of the last $j-1$ special honest slots before slot $i$ that $h$ has heard, the probability $h$ hears the $j\textsuperscript{th}$ from last special honest slot before $i$ is at least $q.$  Therefore, $\Unheard_h[i]$ can be viewed as at most one plus the number of consecutive failures in a sequence of trials, such that each successive trial is successful with probability at least $q$. Thus, $\Unheard_h[i]$ is stochastically dominated by the $\Geom(q)$ distribution, which is the conclusion of (a).

The proof of (b) is similar. Fix $s\geq 1$ and $j\geq 1.$ By the nature of the same renewal process considered in the previous paragraph,  the lifetime that begins at the last renewal point less than or equal to $s$, has the sampled lifetime distribution, equivalent to the sum of two $\Geom(f)$ random variables minus one. Such sampled lifetime distribution is stochastically greater that the typical lifetime distribution, $\Geom(f).$
All the other lifetimes of the renewal process going forwards or backwards from $s$ have the $\Geom(f)$ probability distribution.   Thus, if we consider the renewal process from the perspective of slot $s + T^s_j,$ which is the $j\textsuperscript{th}$ renewal point after slot $s$, the $j\textsuperscript{th}$ lifetime going backwards has the sampled lifetime distribution and all the other lifetimes have the $\Geom(f)$ distribution. Furthermore, these lifetimes are mutually independent. Thus, the same proof as in part (a), with $i$ there replaced by $s + T^s_j$, holds to prove (b).
\end{proof}

\section{Proofs of Lemmas \ref{lem:settle_balance_fork}, \ref{lem:balanced_fork_mu}, and \ref{lem:balanced_fork_equivalence}}  \label{app:deterministic_lemmas}

\settleBalanceFork*
\begin{proof}
By equation \eqref{eq:def_i_settlement_complement}, we know that $\mathcal{E}^c_{i\text{-settlement}}$ implies one of two events. We show that the lemma holds in each case. Let us first consider the case where $\exists \ h_1, h_2 \in \mathcal{I}$ such that $\mathcal{C}^{h_1}_{i}[1:s] \neq \mathcal{C}^{h_2}_{i}[1:s]$. Both $\mathcal{C}^{h_1}_{i}$ and  $\mathcal{C}^{h_2}_{i}$ are tines in $\mathcal{F}_i$. By the definition of $l$, both parties have heard of a special honest broadcast at slot $l$ or later. That is, $\LatestHeard_{h_1}[i] \geq l$ and $\LatestHeard_{h_2}[i] \geq l$. Therefore, both tines are $l$-viable. Finally, since these tines diverge at a block with label $< s$, they must have completely different blocks with labels (timestamps) $s$ onwards. Therefore these tines are $s$-disjoint ($t_1 \nsim_s t_2$). Together, we deduce that $\mathcal{F}_i \vdash \CharString[1:i]$ is an ($s, l$)-balanced fork.
    
Now, consider the case that $\mathcal{C}^{h}_{i}[1:s] \neq \mathcal{C}^{h}_{i+1}[1:s]$
for some $h \in \mathcal{I}.$  Consider the fork $\mathcal{F}_{i+1} \vdash \CharString[1:i+1]$. Let $t_1$ and  $t_2$ be the tines in $\mathcal{F}_{i+1}$ that represent the chains $\mathcal{C}^{h}_{i}$ and $\mathcal{C}^{h}_{i+1}$ respectively. Let $F$ be the directed tree obtained by dropping all blocks with label $i+1$ from $\mathcal{F}_{i+1}$. We now show that $F$ is an ($s, l$)-balanced fork. To prove this, we first note the following properties of $F$.
    \begin{itemize}
        \item $F \vdash \CharString[1:i]$. This follows from the construction of $F$ from $\mathcal{F}_{i+1}$ and $\mathcal{F}_{i+1} \vdash \CharString[1:i+1]$.
        \item $\mathcal{F}_i \sqsubseteq F \sqsubseteq \mathcal{F}_{i+1}$ ($F$ potentially contains some adversarial blocks not in $\mathcal{F}_i$).
        \item If a tine $t \in \mathcal{F}_i$ is $l$-viable in $\mathcal{F}_i$, then $t \in F$ is $l$-viable in $F$.
        \item If $t \in \mathcal{F}_{i+1}$ there is a corresponding tine $\Tilde{t} \in F$ that includes all but possibly the last block of $t$. This is because $t$ may contain at most one block with label $i+1$, which would be the only block not common between $\Tilde{t}$ and $t$. $\text{length}(\Tilde{t}) \geq \text{length}(t) - 1$. 
    \end{itemize}

We know that $t_1$ is an $l$-viable tine in $\mathcal{F}_i$, because it was held by an honest party in slot $i$ and $l \leq \LatestHeard_h[i]$. By the properties of $F$ above, $t_1$ is an $l$-viable tine in $F$. Further, there is a tine $\Tilde{t}_2 \in F$ corresponding to $t_2$. $t_2$ is a tine in $\mathcal{F}_{i+1}$ that is strictly longer than $t_1 \in \mathcal{F}_i$, and therefore $t_2 \in F$ must be at least as long as $t_1 \in F$. Therefore $t_2$ is also an $l$-viable tine in $F$. Lastly, $t_1 \nsim_s t_2$, because they represent chains that diverge prior to slot $s$ (here, $t_1, t_2$ are tines in $\mathcal{F}_{i+1}$). Therefore, $t_1 \nsim_s \Tilde{t}_2$ in the fork $F$. Thus, $F \vdash \CharString[1:i]$ is an ($s, l$)-balanced fork.
\end{proof}

\balanceMargin*
\begin{proof}
(if)
The proof relies on the definitions of margin, reach, reserve and gap (Definitions \ref{def:gap_reserve_reach} and \ref{def:margin}). Suppose $\Margin_s(w) \geq 0$. Then there exists a closed fork $\Bar{F} \vdash w$ such that $\Margin_s(\Bar{F}) \geq 0$. We shall construct $F \vdash w$ such that $\Bar{F} \sqsubseteq F$ and $F$ is $s$-balanced. Note that $\Margin_s(\Bar{F}) \geq 0$ implies $\Bar{F}$ has two tines $\Bar{t}_1$, $\Bar{t}_2$ such that $\Bar{t}_1 \nsim_s \Bar{t}_2$ and reach($\Bar{t}_j$) $\geq 0$, $j \in \{1, 2\}$. (In what follows, any statement with subscript $j$ holds for $j \in \{1, 2\}$). It follows that reserve($\Bar{t}_j$) $\geq$ gap($\Bar{t}_j$).

Recall that reserve($\Bar{t}_j$) are the number of adversarial slots in $w$ whose label is strictly greater than $\ell(\Bar{t}_j)$. This implies we can construct a fork $F \vdash w$ from $\Bar{F}$ by extending each tine $\Bar{t}_j$ by reserve($\Bar{t}_j$) adversarial blocks. Let $t_j$ denote the tine in $F$ extending $\Bar{t}_j$. Then $t_1 \nsim_s t_2$. By the definition of gap, tine $\Bar{t}_j$ is shorter than the longest tine in $\Bar{F}$ by gap($\Bar{t}_j$). Since reserve($t_j$) $\geq$ gap($t_j$), both tines $t_j$ are now longer than the longest tine in $\Bar{F}$. From the third observation made following Definition \ref{def:viable_tine}, and the fact that $\Bar{F}$ is the closure of $F$, we conclude that both $t_j$ are viable in $F$. Thus $F$ is an $s$-balanced fork.

(only if)
For this portion, we work with the definition of viable tines and $s$-balanced forks (Definitions \ref{def:viable_tine} and \ref{def:balanced_forks}).
Let $F \vdash w$ be an $s$-balanced fork. Then there exists two tines $t_1, t_2 \in F$ such that $t_1 \nsim_s t_2$ and they are both viable. Let $\Bar{F}$ be the closure of $F$, and let $\Bar{t}_1, \Bar{t}_2$ be the trimmed versions of $t_1$ and $t_2$ in $\Bar{F}$. It is sufficient to show that
$\Bar{t}_1 \nsim_s \Bar{t}_2$ and  reach($\Bar{t}_1$), reach($\Bar{t}_2$) $\geq 0$. Together, they imply
\begin{align*}
\Margin_s(w) & \geq \Margin_s(\Bar{F}) = \max_{t' \nsim_s t''} \min \left\{\text{reach}(t'), \text{reach}(t'') \right\} \\
& \geq \min \{\text{reach}(\Bar{t}_1), \text{reach}(\Bar{t}_2) \} \geq 0.
\end{align*}

The first point, $\Bar{t}_1 \nsim_s \Bar{t}_2$, follows from the observation after Definition \ref{def:fork_prefixes}. We now show reach($\Bar{t}_j$) $\geq 0$, $j = \{1, 2\}$. First, we note that $\text{length}(t_j) \leq \text{length}(\Bar{t}_j) + \text{reserve}(\Bar{t}_j)$; this follows from the definition of reserve. Rearranging this inequality, we get $\text{reserve}(\Bar{t}_j) \geq \text{length}(t_j) - \text{length}(\Bar{t}_j)$. Second, let $t$ be the longest tine in $\Bar{F}$. By definition, $\text{gap}(\Bar{t}_j) = \text{length}(t) - \text{length}(\Bar{t}_j)$. Third, we note that $t$ is also the longest tine in $F$ that ends in a vertex with a label in $\mathcal{N}_0(w)$. By the definition of viability, $\text{length}(t_j) \geq \text{length}(t), j = \{1, 2\}$. Putting these terms together, we get reach($\Bar{t}_j$) $=$ reserve($\Bar{t}_j$) $-$ gap($\Bar{t}_j$) $\geq \text{length}(t_j) - \text{length}(\Bar{t}_j) - (\text{length}(t) - \text{length}(\Bar{t}_j)) = \text{length}(t_j) - \text{length}(t) \geq 0$.
\end{proof}

\equivalence*
\begin{proof}
(if)
Let $F' \vdash w'$. Let $t$ denote the longest tine in $F$ ending at a block with label in $\mathcal{N}_0(w')$. We create a fork $F \vdash w$ by extending $t$ with a string of special honest nodes corresponding to slots in $\mathcal{N}_0(w[l+1:i])$. If $t'_1 \nsim_{s} t'_2$ are two viable tines in $F'$, then length($t'_j$) $\geq$ length($t$). Since $t$, $t'_1$, and $t'_2$ remain valid tines in $F$, these inequalities holds in $F$ as well. This implies $t'_1$, $t'_2$ are $l$-viable tines in $F$. The property $t'_1 \nsim_{s} t'_2$ trivially extends from $F'$ to $F$. Thus $F \vdash w$ is an $(s, l)$-balanced fork. 

(only if)
Let $F \vdash w$ be an $(s, l)$-balanced fork. We know there exist tines $t_1$ and $t_2 \in F$ such that $t_1 \nsim_s t_2$ and both $t_1$ and $t_2$ are $l$-viable in $F$. Let $t$ be the longest tine in $F$ that ends at a block with a label in $\mathcal{N}_0(w[1:l])$ (such a tine is unique in each $F$). Then $\text{length}(t_1)$, $\text{length}(t_2)$ $\geq$ $\text{length}(t)$. Let $F' \vdash w'$ be a prefix of $F$, obtained by dropping all blocks with labels in $\mathcal{N}_0(w[l+1:i])$ and their descendants. Let $t'_1$ and $t'_2$ be the tines in $F'$ corresponding to $t_1$ and $t_2$. To show $F'$ is an $s$-balanced fork, it is sufficient to show that $t'_1 \nsim_s t'_2$ and $t'_1$ and $t'_2$ are viable tines in $F'$.

The first point, $t'_1 \nsim_s t'_2$, follows from the first observation after Definition \ref{def:fork_prefixes}. Note that $t$ is the longest tine in $F'$ ending at a block with label in $\mathcal{N}_0(w')$. To establish viability, it suffices to show that length($t'_j$) $\geq$ length$(t)$ for $j \in \{1, 2\}$. Note that if any block from tine $t_j$ was dropped to obtain $t'_j$, it must have been at a depth strictly greater than length($t$). This is because any block with a label in $\mathcal{N}_0(w[l+1:i])$ must be at a depth strictly greater than length($t$), by the fifth property of forks (see Definition \ref{def:fork}). Therefore, length($t'_j$) $\geq$ length$(t)$ for $j \in \{1, 2\}$, which is what we wish to prove.
\end{proof}

\section{Proof of Theorem \ref{thm:main}}\label{app:probabilistic_lemmas}

By Lemmas \ref{lem:settlement_necessary} and \ref{lem:chain_quality_necessary}, it suffices to prove the following:
\iftoggle{arxiv}
{
\begin{align}
    \mathbb{P} \left( \Margin_s(\CharString[1:i]) + \Unheard_\mathcal{I}[i] \geq 0 \text{ for some } i \geq s+k \right)
    &\leq p_{\textsf{settlement}} + |\mathcal{I}|p_{\textsf{unheard}}  \label{eq:bound_settlement} \\   \mathbb{P}\left(\Advantage_s(\CharString[1:i]) +  \Unheard_{\mathcal{I}}[i] \geq 0 \text{ for some } i \geq s+ k\right)
    &\leq p_{\textsf{CQ}} + |\mathcal{I}|  \tilde{p}_{\textsf{unheard}} \label{eq:bound_chain_quality}
\end{align}
}
{
\begin{align}
    \mathbb{P} \left(
    \begin{array}{c}
    \Margin_s(\CharString[1:i]) + \Unheard_\mathcal{I}[i] \geq 0 \\
    \text{ for some } i\geq s+k 
    \end{array} \right) \nonumber \\
    \leq p_{\textsf{settlement}}
    + |\mathcal{I}|p_{\textsf{unheard}} \label{eq:bound_settlement}\\
    \mathbb{P}\left(  \begin{array}{c}
    \Advantage_s(\CharString[1:i]) +  \Unheard_{\mathcal{I}}[i] \geq 0 \\ \text{ for some } i \geq s+ k
    \end{array} \right)  \nonumber \\
    \leq p_{\textsf{CQ}} + |\mathcal{I}|  \tilde{p}_{\textsf{unheard}} \label{eq:bound_chain_quality}
\end{align}
}

\subsection{Reduction to compressed time scale}  \label{sec:time_reduction}
We defined compressed time-scale processes in Section \ref{sec:compressed_time_scale}. In this section, we specify events on the compressed time scale implied by the events on the left-hand sides of \eqref{eq:bound_settlement} and \eqref{eq:bound_chain_quality}. First, we establish some notation.

Recall that $\Reach$ denotes both a mapping of strings to $\mathbb{Z}_+$ and the random process $\Reach[i]=\Reach (\CharString [1:i]).$ We define similar random processes for $\Margin_s$ and $\Advantage_s.$ Fix $s\geq 1.$  Then $\Margin_s$ is a mapping of strings to $\mathbb{Z}_+$, where $s$ relates to $s$-disjoint tines.  We now apply this mapping to $\CharString$ and define a random process with the same name: $$\Margin_s[i] \stackrel{\triangle}{=} \Margin_s(\CharString[1:i]).$$
We now define a random process on the compressed time scale based on $\Margin_s$ by using the same value $s$ for both the parameter in defining disjoint tines and the reference slot for the compressed process. Thus, $\CompressedMargin_s[0]=\Margin_s[s]$ and for $j\geq 1,$
\begin{align*}
\CompressedMargin_s[j]  \stackrel{\triangle}{=}
\Margin_s(\CharString[1:s+T_j^s]).
\end{align*}
We define $\Advantage_s$ and $\CompressedAdvantage_s$ similarly.  Finally, recall the processes $\CompressedUnheard_{h,s}$ and $\CompressedUnheard_{\mathcal{I},s}$ from Section \ref{sec:unheard}.

Now, given $k' \geq 1$, consider the following three events:
\begin{align*}
F_0 & = \{ T^s_{k'} > k  \} \\
F_1 & = \{ \CompressedMargin_s[j] + \CompressedUnheard_{\mathcal{I},s}[j] \geq 0 \iftoggle{arxiv}{}{\\& \hspace{10pt}} \text{ for some } j\geq k'   \} \\
F_2 & = \{ \CompressedAdvantage_s[j] + \CompressedUnheard_{\mathcal{I},s}[j] \iftoggle{arxiv}{}{\\& \hspace{10pt}} \geq 0 \text{ for some } j  \geq k' \} 
\end{align*}

We claim that the event on the left-hand side of \eqref{eq:bound_settlement} implies $F_0 \cup F_1$.  The event on the left-hand side of \eqref{eq:bound_settlement} implies that $\Margin_s[i'] +\Unheard_\mathcal{I}[i'] \geq 0$ for some $i' \geq s+k$. The process $\Margin_s$ is constant over intervals of the form $[T^s_j: T^s_{j+1}-1]$ and the process $\Unheard_\mathcal{I}[i]$ is non-increasing over such intervals. So if $j'$ is such that $s+ T_{j'}$ is the last renewal time less than or equal to $i'$, then $\CompressedMargin_s[j']+\CompressedUnheard_{\mathcal{I},s}[j'] \geq 0.$  If $F_0$ does not hold, then $s+ T^s_{k'} \leq s+ k$, implying that $j' \geq k'$, and hence $F_1$ is true. This completes the proof of the claim. Similarly, the event on the left-hand side of \eqref{eq:bound_chain_quality} implies $F_0 \cup  F_2.$  Thus, to prove \eqref{eq:bound_settlement} and \eqref{eq:bound_chain_quality}, it suffices to obtain upper bounds on $\mathbb{P}(F_0 \cup F_1)$ and $\mathbb{P}(F_0 \cup F_2),$ respectively.

The following lemma will be used to help bound $\mathbb{P}(F_0).$
\begin{lemma}  \label{lemma:time_scales}
Suppose $k' = \lceil rkf \rceil$ such that $0 < r < 1.$  Then $\mathbb{P}(T^s_{k'} > k) \leq \exp(-kf(1-r)^2/2).$
\end{lemma}
\begin{proof}
Note that $\{T^s_{k'} > k\} = \{N(\CharString[s+1:s+k])\leq k'-1\},$ and $N(\CharString[s+1:s+k])$ has the binomial distribution with parameters $k$ and $f.$   Thus
\begin{align*}
\mathbb{P}(T^s_{k'} > k) 
& = \mathbb{P}(\mathsf{binom}(k,f) \leq k'-1) \\
&\leq \mathbb{P}(\mathsf{binom}(k,f) \leq rkf ) \leq \exp(-kf(1-r)^2/2).
\end{align*}
where we use the bound $\mathbb{P}(\mathsf{binom}(n,p) \leq rnp)\leq \exp(np(1-r)^2/2).$
\end{proof}

\subsection{On \textsf{Reach}}\label{sec:reach}

In this section, we show that the marginal distribution of $\Reach$ is stochastically dominated by a geometric random variable. This result is used to bound both $\Margin_s$ and $\Advantage_s$ in later sections.

Let $B$ denote the backwards residual lifetime process for the locations of the non-empty slots in $\mathsf{CharString},$  counting from zero. In other words,  $B_t = \min\{i\geq 0: \mathsf{CharString}[t-i]\neq \perp\}.$ 

\begin{lemma}\label{lem:reach_stationary}
The process $(B_t, \mathsf{Reach}[t])$ is a discrete-time Markov process with equilibrium probability mass function given by $\pi(b,r) = f(1-f)^b \left(1-\frac{1-p}{p} \right)  \left( \frac{1-p}{p}\right)^r.$
In other words, under the equilibrium distribution, $B_t$ is independent of $\mathsf{Reach}[t],$ $B_t$ has the
$\Geom(f)-1$ distribution, and $\mathsf{Reach}[t]$ has the $\Geom\left(\frac{1-p}{p}\right)-1$ distribution.
\end{lemma}   

\begin{proof}
The Markov property follows from (i) the recursion \eqref{eq:reach_recursive} for determining $\Reach$ from $\CharString$ and (ii) the renewal structure of $\mathsf{CharString}$ described in Lemma \ref{lem:renewal_prop_CharString}.
The nonzero transition probabilities out of any given state $(b,r) \in \mathbb{Z}_+^2$ are given by  (with $F_b=\mathbb{P}(\Delta \leq b)$):
\begin{align*}
\mathbb{P}((b,r)\to(b+1,r)) &= 1-f  \\  \mathbb{P}((b,r)\to(0,(r-1)_+))&=f\alpha F_b \\  \mathbb{P}((b,r)\to(0,r+1)) &= f(1-\alpha F_b)
\end{align*}
To verify $\pi$ is the equilibrium distribution, it suffices to check that if the state of the process at one time has distribution $\pi$, then in one step of the process, the probability of jumping out of any given state is equal to the probability of jumping into the state.  For a state of the form $(b,r)$ with $b\geq 1$, the probability of jumping into the state is $\pi(b-1,r)(1-f),$ which is equal to $\pi(b,r),$ the probability of jumping out of the state.   For a state of the form $(0,r)$ with $r\geq 1$, the probability of jumping into the state satisfies the following:
\begin{align*}
&\sum_{b=0}^\infty \pi(b,r-1)f(1-\alpha F_b)  + \sum_{b=0}^\infty \pi(b,r+1)f\alpha F_b \\
&~~~  = \pi(0,r) \left[   
\sum_{b=0}^\infty \frac{p}{1-p}(1-f)^b f(1-\alpha F_b)  + 
\sum_{b=0}^\infty  \frac{1-p}{p}(1-f)^b    f\alpha F_b 
\right] \\ & = \pi(0,r),
\end{align*}
where we used the fact $\alpha \sum_{b=0}^\infty  (1-f)^b f F_b  = p.$ Thus, the probability of jumping into the state $(0,r)$ is equal to $\pi(0,r)$, which is the probability of jumping out of state $(0,r).$   It remains to show probabilities of jumping into and out of state (0,0) are the same, but that follows from the fact it is true for all other states.
\end{proof}

\begin{lemma}\label{lem:reach_dist_bound}
For all integers $i \geq 0$,  $\mathbb{P}(\mathsf{Reach}[i] \geq a) \leq \left(\frac{1-p}{p}\right)^a$ for all $a\in \mathbb{R}_+.$
\end{lemma}

\begin{proof}
By the renewal structure of $\mathsf{CharString}$ described in Lemma \ref{lem:renewal_prop_CharString}, the sequence of non-empty slots is a Bernoulli process with parameter $f$, so the distribution of $B_0$ is $\Geom(f)-1.$   The initialization of $\mathsf{Reach}$ is $\mathsf{Reach}[0]=0.$
Consider a comparison system such that $\mathsf{Reach}[0]$ is a random variable independent of  $\mathsf{CharString}$ with the $\Geom\left(\frac{1-p}{p}\right)-1$ distribution.  Then in the comparison system, $((B_t, \mathsf{Reach}[0]): t\geq 0)$ is a stationary Markov process, and in particular,  $\mathsf{Reach}[t]$ has the $\Geom\left(\frac{1-p}{p}\right)-1$  distribution for all $t.$  

Note that, for $\mathsf{CharString}$ fixed, all the variables $((B_t, \mathsf{Reach}[t]): t\geq 0)$ are nondecreasing functions of the initial state $(B_0,\mathsf{Reach}[0])$, as can be readily shown by induction on $t.$  Since the actual initial state of the original system is less than the initial state of the comparison system, it follows that $\mathsf{Reach}[t]$ in the original system is stochastically dominated by the $\Geom\left(\frac{1-p}{p}\right)-1$ distribution, as promised by the lemma.
\end{proof}

\subsection{A bound on a random walk}\label{sec:random_walk}

Let $\textsf{W}$ denote a simple integer valued random walk with a drift $-\epsilon,$   In other words, $\textsf{W}[0] = 0$ and 
\begin{equation}\label{eq:random_walk_defn}
    \textsf{W}[j + 1] = \begin{cases}
    \textsf{W}[j] + 1 & \text{w.p. } \frac{1-\epsilon}{2} \\
    \textsf{W}[j] - 1 & \text{w.p. } \frac{1+\epsilon}{2} \\
    \end{cases}
\end{equation}
Here, $\epsilon$ can be any value in $[-1,1]$, but in our application, $0 < \epsilon < 1.$
The purpose of this section is to prove the following lemma.

\begin{lemma}\label{lem:random_walk_bound}
Let $W[j]$ be a simple random walk defined in \eqref{eq:random_walk_defn}. For any $c < \epsilon$, for any $k \in \mathbb{N}$,
\begin{align} \label{eq:to_prove}
\mathbb{P}(W[j] \geq -cj \text{ for some } j \geq k) \leq 2\exp\left(-k(\epsilon - c)^2/3 \right)
\end{align}
\end{lemma}

\begin{proof}
Let $b>0,$ to be determined below.
Observe that the event on the left-hand side of \eqref{eq:to_prove} is contained
in $G_1\cup G_2$ where $G_1=\{W[k] \geq -ck - b\}$ and
$G_2= \{ \max_{i \geq 0} (W[i+k] - W[k] + ci) \geq b\}.$

Since $W[k] + k\epsilon$ is the sum of $k$ i.i.d. random variables with $0$ mean, each taking values
in an interval of length two, Hoeffding's inequality implies that for any $\delta > 0$,
$ \mathbb{P}\left(W[k] + k\epsilon  \geq k \delta \right) \leq \exp(-k\delta^2/2)$.
Setting  $\delta = \epsilon - c -(b/k)$ yields $P(G_1)  \leq \exp(-k\delta^2/2).$

Let $Y$ be a random variable such that
\begin{equation*}
    Y = \begin{cases}
    ~~ 1 + c & \text{ w.p. } \frac{1 - \epsilon}{2} \\
    -1 + c & \text{ w.p. } \frac{1 + \epsilon}{2},
    \end{cases}
\end{equation*}
and let $Y_1, Y_2, \ldots$ be i.i.d. copies of $Y$.
Kingman's tail bound \cite{Kingman64} is that, for
$\theta^* = \sup\{\theta > 0 : \mathbb{E}\left[e^{\theta Y}\right] \leq 1\},$ 
\begin{align*}
 \mathbb{P}\left(\max_{i \geq 0} \sum_{i' = 1}^{i} Y_{i'} \geq b \right)    \leq e^{-\theta^* b } 
\end{align*}
To obtain a bound on $\theta^*,$ note that
Hoeffding's lemma for bounded random variables \cite{Hoeffding} implies that
$\mathbb{E}\left[e^{\theta (Y-(c-\epsilon))}\right] \leq e^{\theta^2/2}$.  Taking $\theta=-2(c-\epsilon)$ shows that $\mathbb{E}\left[e^{2(\epsilon - c) Y}\right] \leq 1.$
Therefore $\theta^* \geq 2(\epsilon - c),$
Thus, for any $b \geq 0$,
\begin{equation}\label{eq:kingman_bound}
    \mathbb{P}\left(\max_{i \geq 0} \sum_{i' = 1}^{i} Y_{i'} \geq b \right) \leq e^{-\theta^* b } \leq e^{-2(\epsilon - c) b},
\end{equation}
For any $k \in \mathbb{N}$, we note that the random processes
$(\sum_{i' = 1}^i Y_{i'}: i\geq 0)$ and $(W[i + k] - W[k] + ci : i \geq 0)$ have the same
distribution. Therefore, \eqref{eq:kingman_bound} 
implies $\mathbb{P}(G_2) \leq \exp\left(-2(\epsilon - c) b\right).  $

Thus $\mathbb{P}(G_1\cup G_2) \leq   \exp\left(-k\delta^2/2\right) + \exp\left(-2(\epsilon - c) b\right).$
Setting $b = k(\epsilon - c)\left(1-\sqrt{\frac 2 3}\right)$ gives  $\delta^2/2 = (\epsilon - c)^2/3$ and using
$2\left(1-\sqrt{\frac 2 3} \right) \geq 0.367$ yields
\begin{align*}
\mathbb{P}(G_1\cup G_2)& \leq \exp\left(-k(\epsilon - c)^2/3\right) + \exp\left(-k(0.367)(\epsilon - c)^2 \right) \\
&\leq 2\exp\left(-k(\epsilon - c)^2/3\right) 
\end{align*}
which proves the lemma.
\end{proof}

\subsection{On \textsf{Margin} and proof of settlement bound} \label{sec:margin}

We prove bound \eqref{eq:main_settlement_bnd} in this section. By Section \ref{sec:time_reduction}, it suffices to prove $\mathbb{P}(F_0\cup F_1) \leq p_{\textsf{settlement}} + |\mathcal{I}| p_{\textsf{unheard}}.$
Recall that
$\Unheard_{\mathcal{I}}$ is the maximum over the $|\mathcal{I}|$ processes
 $\Unheard_h$ with $h\in\mathcal{I}.$
It thus suffices to prove the following bounds, where $c$ is a constant determined below such that $0 < c < \epsilon$,
$k' = \lceil 3kf/4 \rceil,$ and $h$ denotes an arbitrary special honest user.
\begin{align}
\iftoggle{arxiv}{}{&}\mathbb{P}(T^s_{k'} > k )    +   \mathbb{P}(\CompressedMargin_s[j] \geq -cj  + \frac{ck'} 2  \text{ for some } j \geq k')  \iftoggle{arxiv}{}{\nonumber \\}
    & \leq p_{\textsf{settlement}} \label{eq:bound_margin} \\
\iftoggle{arxiv}{}{&}\mathbb{P}(\CompressedUnheard_{h,s}[j]  \geq \iftoggle{arxiv}{}{~~~} cj - ck'/2   \text{ for some } j \geq k')\iftoggle{arxiv}{}{\nonumber \\}
&\leq p_{\textsf{unheard}}   \label{eq:bound_unheard}
\end{align}
Lemma \ref{lemma:time_scales} with $r=3/4$ yields that $ \mathbb{P}(T^s_{k'} > k ) \leq \exp(-kf/32).$

The recursions \eqref{eq:reach_recursive} and \eqref{eq:margin_recursive}  imply
$\Margin_s[i] = \Reach[i]$ for $1 \leq i \leq s-1,$ and
$\Margin_s[i] \leq \Reach[i]$ for all $i\geq s.$ In particular, $\CompressedMargin_s[0] \leq \CompressedReach_s[0] = \Reach[s]$. 

The following lemma is adapted from \cite{blum2020combinatorics}:
\begin{lemma}\label{lem:margin_ineq2}
For any $s, k \in \mathbb{N}$,
\begin{equation}\label{eq:margin_ineq2}
    \mathbb{P}(\CompressedMargin_s[j] \geq 0 \text{ for some } j \geq k) \leq \exp{(-k \epsilon^3/3)}
\end{equation}
\end{lemma}
\begin{proof}
The lemma is a slight modification of the first corollary at the beginning of Section 6 of \cite{blum2020combinatorics}, which in turn is based on the theorem in that section. We explain why these results can be adapted to our model, and some differences in the form of the bound on the right-hand side of \eqref{eq:margin_ineq2}. In \cite{blum2020combinatorics}, these results are stated for the quantity $\mu_x(y)$, which roughly maps to the quantity $\CompressedMargin_s[j]$. A subtle difference between the two quantities is that $\CompressedMargin_s[j]$ is a metric concerning tines diverging prior to a reference slot $s$ on the \textit{original time scale}, whereas $\mu_x(y)$ corresponds to a reference slot $|x|$ on the \textit{compressed time scale}. Nevertheless, $(\CompressedReach_s[j], \CompressedMargin_s[j])$ satisfy the same recursions as $(\rho(xy), \mu_x(y))$ (see \eqref{eq:margin_recursive} and the two Lemmas in Section 5 of \cite{blum2020combinatorics}), and are `driven' by a $\{0, 1\}$ valued process satisfying the $\epsilon$-martingale condition. Moreover, the initial values satisfy: $\CompressedMargin_s[0] \leq \CompressedReach_s[0] = \Reach[s] \preceq R^*$; the same holds for $(\rho(x), \mu_x(\varepsilon))$ (see Lemmas in Section 5 and 6.2 of \cite{blum2020combinatorics}). The proof of the theorem of \cite{blum2020combinatorics} depends only on the fact that $(\rho(xy), \mu_x(y))$ satisfy these properties, and thus can be adapted to $\CompressedMargin_s$ as well.

The right-hand side of the inequality of the result in \cite{blum2020combinatorics} (the corollary) is stated as $O(1) \exp(-\Omega(k))$, while we use the expression $\exp(-\epsilon^3k/3)$. The difference in the expression comes from two factors. Firstly, the proof of the theorem in Section 6 involves analyzing the (random) time after which $\mu_x(\cdot)$ is negative forever. Put differently, the proof of the theorem actually proves the stronger statement of the corollary in \cite{blum2020combinatorics}. Thus, the bound presented in the corollary can be obtained without a union bound argument. We therefore omit the factor of $O(1)$ of \cite{blum2020combinatorics}. Secondly, the proofs in Section 6 of \cite{blum2020combinatorics} provide exact expressions for the constants in the error exponent. In particular, any bound of the form $\exp(-ak)$ can be used, if $a$ satisfies
\[1 + a < \sqrt{\frac{1}{1+\epsilon}\left(\frac{2}{\sqrt{1-\epsilon^2}} - \frac{1}{1+\epsilon}\right)}\]
The proof of \cite{blum2020combinatorics} concludes with the expression $a = \epsilon^3(1 - O(\epsilon))/2$; however, it can be analytically verified, using the Maclaurin series for $1/\sqrt{1-\epsilon^2}$ and $1/(1+\epsilon)$, that $a = \epsilon^3/3$ satisfies the above inequality for all $\epsilon \in (0, 1)$.
\end{proof}

Let $T$ be a stopping time with respect to $\CompressedMargin_s$, defined as follows:
\begin{equation}\label{eq:margin_stoppingtime}
    T = \min\{j \geq k: \CompressedMargin_s[j] \geq 0\},
\end{equation}
with the convention that the minimum of the empty set is $\infty.$
 Therefore $T < \infty$ is equivalent to the event $\CompressedMargin_s[j] \geq 0$ for some $j \geq k$. From Lemma \ref{lem:margin_ineq2},  $\mathbb{P}(T < \infty) \leq \exp(-k \epsilon^3/3).$

Over the period $[k, T)$, the behavior of $\CompressedMargin_s$ is identical to that of the simple random walk $\textsf{W}$ defined in \eqref{eq:random_walk_defn}. More precisely, writing $N_0-N_1(w)$ as
short for $N_0(w) - N_1(w)$,  for any $j \in \{k, \ldots, T\}$,
\begin{align*}
   & \CompressedMargin_s[j] - \CompressedMargin_s[k] \iftoggle{arxiv}{}{\\
   &} = N_0 - N_1 (\CompressedCharString_s[k+1:j])
\end{align*}
and, as random processes,
\begin{align*}
N_0-N_1 (\CompressedCharString_s[k+1:j]) 
\stackrel{d.}{=} (W[j-k]: j\geq k).  
\end{align*}
If $T = \infty$, $\CompressedMargin_s[k] < 0$. Putting the above facts together, we get the following result due to the union bound and Lemma \ref{lem:random_walk_bound}:
\begin{align*}
    &\mathbb{P}(\CompressedMargin_s[j] \geq -c(j - k) \text{ for some } j \geq 2k) \\ 
    &~~~~~~ \leq \mathbb{P}(T < \infty) + \mathbb{P}(W[j-k] \geq -c(j - k) \text{ for some } j \geq 2k)\\
    &~~~~~~ \leq \exp{(-k\epsilon^3/3)} + 2\exp{(-k(\epsilon - c)^2/3)}
\end{align*}
Replacing $k$ by $k'/2$ in the above equation yields:
\begin{align*}
&\mathbb{P}(\CompressedMargin_s[j] \geq -cj + ck'/2 \text{ for some } j \geq k') \iftoggle{arxiv}{\leq}{\\
 \leq &}\exp{(-k' \epsilon^3/6)} + 2\exp{(-k'(\epsilon - c)^2/6)}
\end{align*}
By Lemma \ref{lem:Uheard_line_bnd},
\begin{align*}
&\mathbb{P}(\textsf{CompressedUnheard}_{h,s}[j]  \geq   cj - \frac{ ck'} 2 \text{ for some } j \geq k')  \\
 = &\mathbb{P}(\textsf{CompressedUnheard}_{h,s}[j]  \geq \frac{ ck'} 2 + c(j-k')  \text{ for some } j \geq k')  \\
 \leq &\left[\frac{1}{(1-q)(1-(1-q)^c)}\right] \exp(-k'cq/2)  
\end{align*}
Combining the bounds in this section shows that \eqref{eq:bound_margin}  and \eqref{eq:bound_unheard}  and thus also \eqref{eq:main_settlement_bnd} hold if
\begin{align*}
  \exp(-kf/32)  + \exp{(-k' \epsilon^3/6)} + 2\exp{(-k'(\epsilon - c)^2/6)} & \leq p_{\textsf{settlement}}  \\
   \left[\frac{1}{(1-q)(1-(1-q)^c)}\right] \exp(-k'cq/2) & \leq  p_{\textsf{unheard}} 
\end{align*}
for some choice of $c.$   Let $c=\epsilon/2$ and use the fact $k' \geq 3kf/4$ to get that the following is sufficient.
\begin{align*}
  \exp(-kf/32)  + \exp{(-kf \epsilon^3/12)} + 2\exp{(-kf\epsilon^2/32)} & \leq p_{\textsf{settlement}}  \\
   \left[\frac{1}{(1-q)(1-(1-q)^{\epsilon/2})}\right] \exp(-kf \epsilon q/8) & \leq p_{\textsf{unheard}} 
\end{align*}
Also, $q \geq p >0.5.$   Thus, \eqref{eq:main_settlement_bnd} holds for
\begin{align*}
   p_{\textsf{settlement}} &=   \exp{(-kf \epsilon^3/12)} + 3\exp{(-kf\epsilon^2/32)}     \\
   p_{\textsf{unheard}}  & =  \left[\frac{2}{ 1-(1/2)^{\epsilon/2}}\right] \exp(-k f \epsilon/16)
\end{align*}

\subsection{On \textsf{Advantage} and proof of chain quality bound} \label{sec:advantage}

We prove bound \eqref{eq:main_quality_bnd} in this section. By Section \ref{sec:time_reduction}, it suffices to prove
$\mathbb{P}(F_0\cup F_2) \leq p_{\textsf{CQ}} + |\mathcal{I}| \Tilde{p}_{\textsf{unheard}}.$
Let $\gamma, r,$ and $c$ be positive constants, to be specified below, such that
$\gamma+\mu < cr < c < \epsilon.$  We use the fact that
$\Unheard_{\mathcal{I}}$ is the maximum over the $|\mathcal{I}|$ processes
 $\Unheard_h$ with $h\in\mathcal{I}.$
It suffices to prove the following bounds, where
$k' = \lceil rkf \rceil,$ and $h$ denotes an arbitrary special honest user.
\begin{align}
\iftoggle{arxiv}{}{&} \mathbb{P}(T^s_{k'} > k )  \iftoggle{arxiv}{}{\nonumber  \\}  +           \iftoggle{arxiv}{}{&}
\mathbb{P}(\CompressedAdvantage_s[j] \geq -cj  + kf (\gamma + \mu)  \text{ for some } j \geq k')  \iftoggle{arxiv}{}{\nonumber \\ }&
\leq p_{\textsf{CQ}} \label{eq:bound_advantage} \\
\iftoggle{arxiv}{}{&}\mathbb{P}(\CompressedUnheard_{h,s}[j]  \geq \iftoggle{arxiv}{}{~~~} cj - kf (\gamma + \mu)   \text{ for some } j \geq k') \iftoggle{arxiv}{}{ \nonumber \\}
&\leq \Tilde{p}_{\textsf{unheard}}   \label{eq:bound_unheard_a}
\end{align}
Lemma \ref{lemma:time_scales} shows that $ \mathbb{P}(T^s_{k'} > k ) \leq \exp(-kf(1-r)^2/2).$
Next, note that \iftoggle{arxiv}
{
\[\CompressedAdvantage_s[j] = N_0 - N_1(\CompressedCharString_s[1:j]) + kf \mu + \Reach[s].\]
}
{
$ \CompressedAdvantage_s[j]$ is equal to:
\[N_0 - N_1(\CompressedCharString_s[1:j]) + kf \mu + \Reach[s].\]
}
Therefore,
\begin{align*}
&\mathbb{P}(\CompressedAdvantage_s[j] \geq -cj  + kf (\gamma + \mu) \text{ for some } j \geq k') \\
\leq  & \mathbb{P}(N_0 - N_1(\CompressedCharString_s[1:j]) \geq - cj \text{ for some } j \geq k')  \iftoggle{arxiv}{}{\\
&~~~~~~ }   + \mathbb{P}(\Reach[s] \geq kf \gamma )
\end{align*}
Lemma \ref{lem:CCS} implies that for any $s \in \mathbb{N}$ and any $j \geq 2$, 
\[N_0 - N_1(\CompressedCharString_s[2:j]) \stackrel{d.}{=} \textsf{W}[j] - \textsf{W}[1] \]
where $\textsf{W}$ is a simple random walk as defined in \eqref{eq:random_walk_defn}. Moreover, $N_0(\CompressedCharString_s[1:j]) - N_1(\CompressedCharString_s[1:j])$
stochastically dominates $\textsf{W}$, because its value at $j=1$ is one with probability greater than $p.$   By Lemma \ref{lem:random_walk_bound},
\iftoggle{arxiv}
{
\begin{equation*}
    \mathbb{P}(N_0 - N_1(\CompressedCharString_s[1:j] \geq - cj \text{ for some } j \geq k') \leq 2\exp(-k'(\epsilon - c)^2/3)
\end{equation*}
}
{
\begin{multline*}
    \mathbb{P}(N_0 - N_1(\CompressedCharString_s[1:j] \geq - cj \text{ for some } j \geq k') \\ \leq 2\exp(-k'(\epsilon - c)^2/3)
\end{multline*}
}
We next bound $\Reach[s]$ as follows:
\begin{align*}
    \mathbb{P}(\Reach[s] \geq kf \gamma  ) 
    &\leq \left(\frac{1 - p}{p}\right)^{kf \gamma } \quad \text{by Lemma \ref{lem:reach_dist_bound}} \\
    &= \left(\frac{1 + \epsilon}{1 - \epsilon}\right)^{-kf \gamma}\quad \text{by Definition \ref{def:eps_honest_maj}} \\
    &\leq \exp(-2kf \gamma \epsilon) ,
\end{align*}
where the final step follows because $\log((1+\epsilon)/(1-\epsilon)) \geq 2 \epsilon$  for   $ \epsilon \in [0, 1)).$   By Lemma \ref{lem:Uheard_line_bnd},
\iftoggle{arxiv}
{
\begin{align*}
&\mathbb{P}(\textsf{CompressedUnheard}_{h,s}[j]  \geq   cj - kf (\gamma + \mu) \text{ for some } j \geq k')  \\
= &\mathbb{P}(\textsf{CompressedUnheard}_{h,s}[j]  \geq  kf(cr-\gamma-\mu) + c(j-k') \text{ for some } j \geq k')  \\
 \leq  &\left[\frac{1}{(1-q)(1-(1-q)^c)}\right] \exp(-kf(cr-\gamma-\mu)q)  
\end{align*}
}
{
\begin{align*}
&\mathbb{P}(\textsf{CompressedUnheard}_{h,s}[j]  \geq   cj - kf (\gamma + \mu) \text{ for some } j \geq k')  \\
&= \mathbb{P}(\textsf{CompressedUnheard}_{h,s}[j]  \geq  kf(cr-\gamma-\mu) + c(j-k') \\
&\hspace{1in} \text{ for some } j \geq k')  \\
& \leq    \left[\frac{1}{(1-q)(1-(1-q)^c)}\right] \exp(-kf(cr-\gamma-\mu)q)  
\end{align*}
}
Combining the bounds in this section shows that  \eqref{eq:bound_advantage}  and \eqref{eq:bound_unheard_a}  and thus also \eqref{eq:main_quality_bnd} holds if
\begin{align*}
 & \exp(-kf(1-r)^2/2) + 2\exp(-k'(\epsilon - c)^2/3) + \exp(-2kf \epsilon \gamma ) \leq p_{\textsf{CQ}}  \\
&   \left[\frac{1}{(1-q)(1-(1-q)^c)}\right] \exp(-kf(cr-\gamma - \mu)q) \leq  \tilde{p}_{\textsf{unheard}} 
\end{align*}
for some choice of $\gamma$, $c$, and $r.$  Select these constants so that the five values,
$\mu, \gamma+ \mu, cr, c, \epsilon,$ form an arithmetic sequence, i.e. the consecutive values each differ by $\gamma = \frac{\epsilon - \mu} 4.$   Observe that
$1 -  r =  1 - \frac{c - \gamma } c = \gamma/c \geq  \gamma$  and use the fact $k' \geq rk.$  So it is sufficient that:
\begin{align*}
&  \exp(-kf \gamma^2/2)  + 2\exp(-kf \gamma^2 /3) + \exp(-2kf \gamma  \epsilon ) \leq p_{\textsf{CQ}}  \\
&  \left[\frac{1}{(1-q)(1-(1-q)^c)}\right] \exp(-kf \gamma q) \leq  \tilde{p}_{\textsf{unheard}} 
\end{align*}
Also, $q \geq p >0.5,$ $c\geq \epsilon/2,$ and $\gamma = \frac{\epsilon - \mu} 4$  can be used.   Thus, \eqref{eq:main_quality_bnd} holds for
\begin{align*}
   p_{\textsf{CQ}} &= 4 \exp(-kf (\epsilon-\mu)^2 /48)   \\
   \tilde{p}_{\textsf{unheard}} & =   \left[\frac{2}{ 1-(1/2)^{\epsilon/2}}\right] \exp(-kf(\epsilon- \mu) /8)
\end{align*}

\end{document}